\documentclass[11pt]{article}

\usepackage{a4wide}

\usepackage[english]{babel}

\usepackage{amsfonts}
\usepackage{amsmath} \usepackage{tikz,graphicx,subfigure,overpic,verbatim}
\usepackage{color} \usepackage{amssymb} \usepackage{amsthm}
\usepackage{hyperref}
\usepackage[fixlanguage]{babelbib}
\usetikzlibrary{arrows,decorations.markings}
\def\C{\mathbb{C}}
\def\N{\mathbb{N}}

\def\essinf{\mathop{\mathrm{ess inf}}}

\def\Tr{\mathop{\mathrm{Tr}}\nolimits}
\def\Re{\mathop{\mathrm{Re}}\nolimits}
\def\Im{\mathop{\mathrm{Im}}\nolimits}

\def\eps{\varepsilon}

\def\bbR{\mathbb{R}}
\def\bbD{\mathbb{D}}
\def\bbZ{\mathbb{Z}}
\def\bbN{\mathbb{N}}
\def\bbT{\mathbb{T}}
\def\calC{\mathcal{C}}
\def\calD{\mathcal{D}}

\def\calE{\mathcal{E}}

\def\T{\operatorname{T}}
\def\Ha{\operatorname{H}}
\def\DT{\operatorname{DT}}
\def\QT{\operatorname{QT}}

\def\fr{\stackrel{f.e.}{=}}

\newcommand{\bdnot}{{\boldsymbol{0}}}

\newcommand{\tb}[1]{{#1}^*}

\numberwithin{equation}{section}

\newtheorem{theorem}{Theorem}[section]
\newtheorem{lemma}[theorem]{Lemma}
\newtheorem{proposition}[theorem]{Proposition}

\newtheorem{corollary}[theorem]{Corollary}

\theoremstyle{definition}

\theoremstyle{remark}

\newtheorem{remark}{Remark}

\numberwithin{remark}{section}

\title{Relative  Szeg\H{o} asymptotics for Toeplitz determinants}

\author{Maurice Duits\footnote{Department of Mathematics, Royal Institute of Technology (KTH), Stockholm Lindstedsv\"agen 25, SE-10044, Sweden. Email: duits@kth.se.  Supported  by the  Swedish Research Council (VR) Grant no.\ 2012-3128.}  \and
Rostyslav Kozhan\footnote{Department of Mathematics, Uppsala University, Box 480, 75106 Uppsala, Sweden. Email: rostyslav.kozhan@math.uu.se.  Partially supported by the grant KAW 2010.0063 from the Knut and Alice Wallenberg Foundation.}
}

\begin{document}

\maketitle

\begin{abstract}
We study the asymptotic behavior, as $n \to \infty$,  of ratios  of Toeplitz determinants $D_n({\rm e}^h {\rm d}\mu)/D_n({\rm d}\mu)$  defined by a measure $\mu$ on the unit circle and a sufficiently smooth function  $h$.  The approach we follow is based on
the theory of orthogonal polynomials.
We prove that the second order asymptotics depends on  $h$ and only a few Verblunsky coefficients associated to $\mu$. As a result, we establish a relative version of the Strong Szeg\H{o} Limit Theorem for a wide class of measures
$\mu$ with essential support on a single arc. In particular, this allows the measure to have a singular component within or outside of the arc.
\end{abstract}

\section{Introduction}
Let $\mu$ be a finite Borel measure on the unit circle $\mathbb T = \{z \in \C \mid |z|=1\}$ with infinitely many points in its support. The Toeplitz matrix of size $n \in \N$ associated to $\mu$ is defined as the matrix
$$T_n({\rm d} \mu)= (c_{j-k})_{j,k=1}^n,$$
with
$$c_k = \int_{\mathbb T} z^{-k}  {\rm d}\mu(z).$$
We then denote its determinant by $D_n({\rm d} \mu)$, i.e.,
$$D_n({\rm d} \mu)= \det T_n({\rm d} \mu).$$
The purpose of this paper is to study  the asymptotic behavior  of the ratio
\begin{equation}\label{eq:ratio}
\frac{D_n({\rm e}^h {\rm d} \mu)}{D_n({\rm d} \mu)}, \quad \text{ as } n \to \infty,
\end{equation}
where $h: \mathbb T \to \C$ is a function on the unit circle on which we will impose certain smoothness conditions.  Since we are concerned with the \emph{ratio} of two Toeplitz determinants, we speak of \emph{relative} asymptotics.

Toeplitz matrices and their determinants appear in numerous places in mathematics and mathematical physics and are therefore very well-studied in the literature. Statistical mechanics  has proved to be a particularly rich  source, as various quantities in models of interest can be reduced to studying the asymptotic behavior of special Toeplitz determinants. A classical example is spin-spin correlations for the Ising model  leading to a Toeplitz determinant with a Fisher-Hartwig symbol as originally shown by Kaufman and Onsager. The asymptotic study of Toeplitz determinants for Fisher-Hartwig symbols has witnessed dramatic progress  in the last two decades, e.g. \cite{DIKannals,EhrhF}, and  we refer the interested reader to \cite{DIK,KrasovskyRev} for  recent reviews and a good source for further references.  Random matrix theory is another example of a discipline where Toeplitz matrices arise and  we will discuss a particular application of asymptotics for the ratio \eqref{eq:ratio}  in more detail below. Toeplitz matrices also relate naturally to the spectral theory of unitary operators and orthogonal polynomials on the unit circle: indeed, using the representation~\eqref{eq:heine}, one can see that $D_n({\rm d} \mu)$ can be rewritten in terms of the leading coefficients of orthonormal polynomials. 
Relative asymptotics of (the leading coefficients of) orthonormal polynomials is an important topic of investigation in spectral theory. We briefly discuss related Nevai's conjecture in Section~\ref{ssWeak} below.

One of the most celebrated results  on the asymptotic behavior for Toeplitz determinants is the \emph{Strong Szeg\H{o} Limit Theorem}: if ${\rm d} \mu=  \frac{{\rm d} \theta}{2 \pi} $ where ${\rm d}  \theta$ is the arclength measure on the circle and $h$ is a sufficiently smooth function, then
\begin{equation}
\label{eq:Szego}
 D_n\left({\rm e}^h \frac{{\rm d}\theta}{2 \pi}\right)= {\rm e}^{n h_0+  \sum_{k=1}^\infty k h_k h_{-k}} (1+o(1)).
\end{equation}
as $n \to \infty$,
where
\begin{equation}
\label{eq:fourier}
h_k= \frac{1}{2 \pi } \int_0^{2  \pi} h({\rm e}^{{\rm i } \theta}) {\rm e}^{-{\rm i} k \theta} {\rm d} \theta.
\end{equation}
The Strong Szeg\H{o} Limit Theorem has a long history with various applications to mathematical physics. It was first proved \cite{Szego3} by Szeg\H{o} in 1952 under stronger conditions on $h$ and further developed in, e.g.,   \cite{Baxter, GolIbr,Ibragimov,JohSze,WidomSzego} and many others. See \cite[Ch. 6]{OPUC1} for a collection of different proofs and  \cite{DIK,KrasovskyRev} for   excellent surveys on the recent progress on the topic.

Since $D_n( {\rm d} \theta/2 \pi)=1$ we see that \eqref{eq:Szego} also provides the asymptotic behavior of the ratio in \eqref{eq:ratio}. Written as a limit of the ratio,  the Strong Szeg\H{o} Limit Theorem tells us how the asymptotic behavior of the Toeplitz determinant changes when we perturb the arclength measure  by multiplying it with a sufficiently smooth density.  In this paper we  study the same question for more general measures $\mu$.  From known results in the literature, such as the Strong Szeg\H{o} Limit Theorem and extensions to symbols with, e.g., Fisher-Hartwig type of singularities,  it is reasonable to expect that there can only exist an analogue to \eqref{eq:Szego} if $h$ is sufficiently smooth. But it is a priori less clear what regularity assumptions are natural on the measure $\mu$. Somewhat surprisingly, the class of measures for which we prove an equivalent of \eqref{eq:Szego} includes measures that have a non-trivial singular component. Fisher-Hartwig symbols are also included in our results and we include a short discussion (cf. Section~\ref{sec:FH}) on how it explains some of the terms in the asymptotic expansion that is known in the literature.

Another, perhaps more concrete,
motivation  for studying relative asymptotics   \eqref{eq:ratio} comes from Random Matrix Theory or Coulomb gases on the circle.  We refer to the review paper \cite{Diac} for more details and background.  The starting point is that the Toeplitz determinant can be represented as a multiple integral,
\begin{equation} \label{eq:heine}
D_n({\rm e}^h {\rm d} \mu)
= \frac{1}{n!} \int_{\mathbb T} \cdots \int_{\mathbb T}  {\rm  e}^{\sum_{j=1}^n h(z_j)} \prod_{1 \leq j <k \leq n} |z_j-z_k|^2  {\rm d} \mu(z_1) \cdots {\rm d} \mu(z_n).
\end{equation}
By taking $h(z)= {\rm i} t f(z)$,  we see that we can thus write
$$\frac{D_n({\rm e} ^{{\rm i}t f} {\rm d} \mu)}{D_n({\rm d} \mu)}=\mathbb E \left[{\rm e}^{{\rm i} t X_n(f)}\right],$$
where $X_n(f)=\sum_{j=1}^nf(z_j)$ is the linear statistic defined by $f$ and the expectation is taken with respect to the probability measure on $\mathbb T^n$ proportional to
\begin{equation}
\label{eq:coulombgas}
\prod_{1 \leq j <k \leq n} |z_i-z_j|^2  {\rm d} \mu(z_1) \cdots {\rm d} \mu(z_n).\end{equation}
Note that if ${\rm d}\mu$ is the arclength measure then this probability measure describes the joint probability for the eigenvalues of an $n \times n$ unitary matrix taken randomly with respect to the Haar measure, i.e. a CUE matrix. In  the more general situation the eigenvalues are also influenced by the environment $\mu$.

Linear statistics are  natural and well-studied objects for random point processes \cite{JohDet}. A consequence of the Strong Szeg\H{o} Theorem \eqref{eq:Szego} is that  smooth linear statistics of the CUE (i.e. ${\rm d} \mu= {\rm d} \theta/ 2\pi$) obey a Central Limit Theorem. It is expected that such a Central Limit Theorem is not special for the CUE, but should hold under fairly mild conditions on the measure $\mu$. Indeed, similar results have been rigorously verified in many models in Random Matrix Theory and Integrable Probability by various authors. We single out \cite{BDjams} where one of us  together with Breuer  proved a universal Central Limit Theorem for biorthogonal ensembles on the real line. The methods developed in  \cite{BDjams} are an important inspiration to us for the present paper. We continue on this development and extend the approach to deal with measures on the circle and obtain universal asymptotics for \eqref{eq:ratio} under mild conditions on the measure $\mu$.

The main results in the paper are Theorems \ref{thm2}, \ref{thm1} and \ref{thm:rightlimit}. Roughly speaking, the main conclusion of Theorem \ref{thm2} is that  the second order asymptotics is universal and only depends on certain properties of the measure, namely, the right limits of the Verblunsky coefficients (whose definition we recall in the next section).   It allows us to divide the measure into classes and conclude that two measures in the same class have the same second order asymptotics. Each class has its own limiting behavior as stated in Theorem \ref{thm:rightlimit}. However, the limiting expression in general is not explicit and we  compute a more concrete form of the limit for a special important class in Theorem \ref{thm1}. This class is defined by the L\'opez condition and thus contains all measures $\mu$ for which (1) the essential support is a single arc  and (2) the absolutely continuous part has full support on that arc.

But before we state our main results in full generality, let us first illustrate them by discussing two special corollaries. First let $\mu$ be a measure on $\mathbb T$ for which the absolutely continuous part satisfies ${\rm d} \mu/{\rm d} \theta>0$ for almost every $\theta \in [0,2\pi)$. Note that $\mu$ may have an arbitrary singular part.   Then we will  prove (cf. Corollary \ref{cor:alpha0}) that
\begin{equation}
\label{eq:Szego1}
 \lim_{n\to \infty} \frac{D_n\left({\rm e}^h {\rm d}\mu\right)}{D_n\left({\rm d} \mu \right)}{\rm e}^{-n L_n(h)}= {\rm e}^{\sum_{k=1}^\infty k h_k h_{-k}},
\end{equation}
for sufficiently smooth $h$.
The term $L_n(h)$ is linear in $h$ and can be expressed in terms of the orthogonal polynomials with respect to $\mu$. Let  $\Phi_n(z)$ be the unique monic  polynomial in $z$ of degree $n$ such that
\begin{equation}
\label{szegoRec}
\int_{\mathbb T} \Phi_n(z) \bar z ^k {\rm d} \mu(z)= 0, \qquad k=0,1,\ldots,n-1.
\end{equation}
Then
\begin{equation}\label{eq:defLn}
L_n(h)= \frac1n \int h(z) K_n(z,z) {\rm d} \mu(z),
\end{equation}
where $K_n$ is the reproducing kernel defined by$$K_n(z,w)= \sum_{j=0}^{n-1} \frac{\Phi_j(z) \overline{\Phi_j(w)} } {\|\Phi_j\|_2^2 }.$$
Another corollary of our results is the following. If the essential support of $\mu$ (i.e. the support of $\mu$ with isolated points removed) is an arc $\{{\rm e} ^{{\rm i } \theta}\mid \theta \in [\phi, 2\pi-\phi]$ and ${\rm d} \mu / {\rm d} \theta >0$ on that arc, then (cf. Corollary \ref{cor:alpha})
\begin{equation}
\label{eq:Szego2}
 \lim_{n\to \infty} \frac{D_n\left({\rm e}^h {\rm d}\mu\right)}{D_n\left({\rm d} \mu \right)}{\rm e}^{-n L_n(h)}= {\rm e}^{Q(h)},
\end{equation}
for sufficiently smooth $h$. Here $Q(h)$ is a quadratic form that is entirely determined by the endpoints of the arc. The precise explicit description will be given later in~\eqref{eq:defQa}.

Both \eqref{eq:Szego1} and \eqref{eq:Szego2}  are examples of the following general problem. For a measure $\mu$   and a sufficiently smooth~$h$ consider the function~$\Psi_n$ defined by
 \begin{equation}\label{eq:strongasympterm}
\Psi_n(h,\mu)= \frac{D_n(e^h {\rm d} \mu)}{D_n({\rm d}\mu)} e^{-  \int h(z) K_n(z,z) {\rm d} \mu(z) },
 \end{equation}
 and find its asymptotic behavior as $n \to \infty$.  This is the central question of the paper. As the above examples show, the limiting behavior is universal in the sense that it only depends on certain properties of the measure. In the examples it is  the essential support, but we will pose even weaker conditions.
 \subsubsection*{Acknowlegdements}
We thank Jonathan Breuer, Kurt Johansson and Igor Krasovsky for fruitful discussions  and  Percy Deift for his comments that helped improving the presentation of the paper.  We are very  grateful to an anonymous referee for pointing out a mistake in the proof of Proposition 2.9 in an earlier manuscript and his/her suggestion to restrict to sectorial symbols.
 \section{Statement of results}
 In this Section we will state our main results.
\subsection{Verblunsky coefficients and a comparison result}

The approach we follow in this paper is based on the Verblunsky coefficients associated to $\mu$. The orthogonal polynomials $\Phi_n$ satisfy the well-known recurrence relation
\begin{equation}
\label{eq:rec}
z \Phi_n(z)= \Phi_{n+1}(z)-\bar{\alpha}_n \Phi^*_n(z),
\end{equation}
where $\alpha_n\in \bbD \equiv\{z:|z|<1\}$, and $\Phi_n^*(z)= z^n \overline{ \Phi_n(1/\bar{z})}$ is the reciprocal polynomial. We will refer to the recurrence coefficients $\alpha_n$ as the Verblunsky coefficients of the measure $\mu$. Conversely, for each sequence $\{\alpha_n\}_{n=0}^\infty$, $\alpha_n\in\bbD$, there exists a unique probability measure $\mu$ with $\{\alpha_n\}_{n=0}^\infty$ as its Verblunsky coefficients. We refer the reader to~\cite{OPUC1,OPUC2} for  proofs of these and other results from the theory of orthogonal polynomials on the unit circle.

Since the Verblunsky coefficients determine the measure $\mu$ uniquely, it is natural to turn to the question: under what conditions on $\alpha_j$ do we have an analogue of \eqref{eq:Szego}? As we will see shortly, the Verblunsky coefficients are a very useful tool in the asymptotic analysis since $\Psi_n$ depends mostly on very few coefficients and only weakly on the others. This observation was inspired by the recent papers \cite{BDjams} where a similar approach based on the Jacobi operator for a measure on $\mathbb R$  turned out to be successful in the context of Central Limit Theorems for linear statistics for the Orthogonal Polynomial Ensembles. In \cite{BDmeso} it was also applied to mesosopic scale statistics and  two-dimensional systems of non-colliding processes in \cite{duits}. In the present paper, we further develop and extend these ideas in the context of Toeplitz determinants.

 In the exponent on the right-hand side of \eqref{eq:Szego} there is a term that is linear in $h$ that grows linearly, as $n \to \infty$, and a quadratic term that is constant in $n$. For general measures $\mu$ the first term is replaced by  \eqref{eq:defLn}.  This term depends on all of the Verblunsky coefficients, which is easy to verify by taking a Laurent polynomial $h$ and iterating the recurrence \eqref{eq:rec}. However, the key observation in this paper is that  the quadratic term will be replaced by a term that depends strongly only  on the  Verblunsky coefficients around the $n$-th position  and weakly on the others. This also shows that it is \emph{universal} since it is the same for all different Verblunsky sequences, and hence different measures, for which the relevant coefficients have the same asymptotic behavior.   This is formulated  more precisely  in the following comparison principle which is  our first main result.

\begin{theorem} \label{thm2}
Let $\{\alpha_k\}_{k \in \mathbb N}$ and $\{\tilde  \alpha_k\}_{k \in \mathbb N}$ be the Verblunsky coefficients corresponding to two measures $\mu$ and $\tilde \mu$.  Assume that there exists a subsequence $\{n_j\}_{j \in \mathbb N}$ of $\mathbb N$ such that, for any $k \in \mathbb Z$,
\begin{equation}\label{eq:comparison} \lim_{j\to \infty} \left( \alpha_{n_j+k} - \tilde \alpha_{n_j+k} \right)= 0.
\end{equation}
Then, with $\Psi_n$ as defined in \eqref{eq:strongasympterm},
\begin{equation}
\lim_{j \to \infty} \left(  \Psi_{n_j}(h,\mu)-\Psi_{n_j}(h,\tilde \mu)\right) =0,
\end{equation}
for all $h \in  \mathfrak B _{\frac12}$, where
\begin{equation} \label{eq:sobolevspace}
\mathfrak B _{\frac12}= \{ h: \mathbb T \to \mathbb C \mid  \|h\|_{ \mathfrak B _{\frac12}}:=\sum_{k \in \mathbb Z} \sqrt{1+ |k|}  |h_k| <\infty\}.
\end{equation}
\end{theorem}
The proof of this Theorem is given in Section \ref{sec:proofthm2}.

We will discuss this class of functions $\mathfrak B_{\frac12}$ more thoroughly in Section~3. This class of functions was also used by Baxter \cite[Th. 3.2]{Baxter} in his proof of the Strong Szeg\H{o} Limit Theorem. A useful property of $\mathfrak B_{\frac12}$ is that it is a Banach algebra.

Note that Theorem \ref{thm2} shows the universality of $\Psi_n$: without specifying the limit (or even establishing the existence of a limit) we show that the asymptotic behavior is invariant under small perturbations of the Verblunsky coefficients. Moreover, Theorem \ref{thm2} allows us to consider general classes of comparable  measures (in the sense of \eqref{eq:comparison}) that have a special member  for which we can compute the asymptotic behavior explicitly. In this paper we compute some examples, which we will discuss next.

\subsection{Special cases: measures supported on arcs}

If $\tilde \alpha_n\equiv 0$ then ${\rm d}  \tilde \mu= {\rm d} \theta/2\pi$ and we readily obtain the following generalization of the Strong Szeg\H{o} Limit Theorem~\eqref{eq:Szego} by combining  it with Theorem~\ref{thm2}.

\begin{proposition}\label{prop:alpha0}
Let $\mu$ be such that along a subsequence $\{n_j\}_{j \in \mathbb N}$ we have
$$\lim_{j\to \infty} \alpha_{n_j+k} =0, \qquad \text{for all } k \in \mathbb Z.$$
Then
\begin{equation}\label{eq:generalSSLT1cor0}
\lim_{j\to \infty}\Psi_{n_j}(h,\mu)= {\rm e}^{\sum_{k=1}^\infty k h_k h_{-k}},
\end{equation}
 for $h \in \mathfrak B_{\frac12}$.
\end{proposition}

The question arises when the condition in this proposition is satisfied. If $\alpha_n \to 0$ as $n\to\infty$ then by the Weyl theorem on compact perturbations (see, e.g.,~\cite[Sect. 1.4.15]{OPUC1}), $\sigma_{ess}(\mu) = \mathbb T$ (by $\sigma_{ess}$ we denote the essential support of $\mu$ which equals to the support of $\mu$ with the isolated points removed). In the converse direction we have Rakhmanov's theorem~\cite{Rakh} which states that if $\tfrac{{\rm d}\mu}{{\rm d} \theta} > 0$ for almost every $\theta\in [0,2\pi)$ then $\lim_{n\to\infty} \alpha_n =0$. Note that $\mu$ here may have an arbitrary singular part.    This gives the following corollary.
\begin{corollary} \label{cor:alpha0}
Let $\mu$ be such that  $\tfrac{{\rm d}\mu}{{\rm d}\theta} > 0$ for almost every $\theta\in[0,2\pi).$ Then  \begin{equation}\label{eq:generalSSLT1cor00}
\lim_{n\to \infty}\Psi_{n}(h,\mu)= {\rm e}^{\sum_{k=1}^\infty k h_k h_{-k}},
\end{equation}
 for $h \in \mathfrak B_{\frac12}$.
\end{corollary}

Proposition \ref{prop:alpha0} and its Corollary \ref{cor:alpha0} rely on the known asymptotics coming from the Strong Szeg\H{o} Limit Theorem \eqref{eq:Szego} for $\alpha_n \equiv 0$.  In this paper we will also prove its analogue also for the special  case $\alpha_n \equiv \alpha$.  By applying Theorem \ref{thm2} we can thus also deal with the situation $\alpha_{n_j} \to \alpha$.  This is the next result we discuss.

In case $\alpha_n \equiv \alpha$  the measure $\mu_\alpha$ is supported on an arc (apart from a possible point mass)
\begin{equation}\label{eq:arc}
\Gamma_\phi = \left\{ {\rm e}^{{\rm i} \omega }\mid \omega \in [\phi,2 \pi-\phi)\right\},
\end{equation}
where
\begin{equation}\label{eq:phi}
\phi= 2 \arcsin |\alpha|.
\end{equation}
We do not need the explicit form of the measure, but for completeness we present it in the Appendix. Now consider the map $\theta \mapsto \omega$ defined by
\begin{equation}\label{eq:stretching} \omega=2 \arccos  \left( \rho \cos (\theta/2)\right),
\end{equation}
with $\rho= \sqrt{1-|\alpha|^2}$, for $\theta \in [0,2 \pi)$. Then $\theta \mapsto \omega$ establishes a 1-to-1 correspondence between the unit circle and the arc $\Gamma_\phi$.

Then, for  $h \in \mathfrak B_{\frac12}$ we write $$h({\rm e}^{{\rm i} \theta})= a_0+2\sum_{j=1}^\infty a_j \cos j \theta + 2\sum_{j=1}^\infty b_j \sin j \theta, $$
for some $a_j$ and $b_j$, and define
\begin{align}
\label{eq:Ah} \mathcal{A}^h({\rm e}^{{\rm i} \theta})&= a_0+2\sum_{j=1}^\infty a_j  \cos  j \omega,\\
\label{eq:Bh} \mathcal{B}^h({\rm e}^{{\rm i} \theta} )&=2 (\sin\tfrac\theta2 + |\alpha| \cos\tfrac\theta2)  \sum_{j=1}^\infty b_j  \frac{\sin  j\omega}{\sin \frac{\omega}{2}},
\end{align}
where $\omega=\omega(\theta)$ is defined by \eqref{eq:stretching} , $\theta \in [0,2\pi)$. Note that $\mathcal A^h$ and $\mathcal B^h$ only depend on the values of $h$ on the arc $\Gamma_\phi$. Moreover, $\mathcal A^h$ and $\mathcal B^h$ are determined by the even and odd parts of $h$, respectively.

Finally, we define
\begin{equation}\label{eq:defQa}
Q_\alpha (h) =    \sum_{j=1}^\infty j \mathcal{A}^h_j \mathcal{A}^h_{-j}+ \sum_{j=1}^\infty j \mathcal{B}^h_j \mathcal{B}^h_{-j},
\end{equation}
where $\mathcal{A}^h_j$ and $\mathcal{B}^h_j$ are the $j$-th Fourier coefficients \eqref{eq:fourier} of $\mathcal{A}^h$ and $\mathcal{B}^h$. We will prove that $Q_\alpha(h)$ is indeed well-defined for $h \in \mathfrak B_{\frac12}$, see Lemma \ref{lem:final}.

The following is a generalization of the Strong Szeg\H{o} Limit Theorem \eqref{eq:Szego}.

\begin{theorem}\label{thm1}
Let $\mu$ be a measure on $\mathbb T$ such that its Verblunsky coefficients satisfy
\begin{equation}\label{eq:subslimit}
\lim_{j\to \infty}\alpha_{n_j+k}=\alpha, \qquad \text{for all }  k \in \mathbb Z,
\end{equation}
for some $|\alpha|<1$ and subsequence $\{n_j\}_{j \in \mathbb N}$ of $\mathbb N$. Then
\begin{equation}\label{eq:generalSSLT1}
\lim_{j\to \infty} \Psi_{n_j}(h,\mu)= {\rm e}^{Q_\alpha(h)},
\end{equation}
for  $h \in  \mathfrak B _{\frac12}$.
\end{theorem}
The proof of this Theorem will be given in Section \ref{sec:proofthm1}.

In the special case $\alpha_n \to \alpha$  then it is true that  $\sigma_{ess}(\mu) =\Gamma_\phi$, where $\phi$ is as above. In fact, this holds under weaker assumptions, namely
\begin{align}
\label{L1}
\lim_{n\to\infty} |\alpha_n| &= |\alpha|, \\
\label{L2}
\lim_{n\to\infty} \frac{\alpha_{n+1}}{\alpha_n} &= 1,
\end{align}
which are called  the L\'{o}pez conditions  (see~\cite[Thm. 4.3.8]{OPUC1}). 
\begin{proposition} \label{prop:alpha}
Let $\mu$ be such that \eqref{L1}--\eqref{L2} are satisfied for some $|\alpha|<1$. Then
\begin{equation}\label{eq:generalSSLT1cor1}
\lim_{n\to \infty} \Psi_{n}(h,\mu) = {\rm e}^{Q_\alpha(h)},
\end{equation}
 for $h \in \mathfrak B_{\frac12}$.
\end{proposition}
The proof of this proposition, which is basically a corollary to Theorem \ref{thm1}, will be given in Section \ref{sec:proofthm1}.

The analogue of the Rakhmanov theorem (due to Bello--L\'{o}pez~\cite{BelloLopez}, and improved by Simon~\cite[Thms. 9.9.1 and 13.4.4]{OPUC2}) says that if a probability measure $\mu$ satisfies $\sigma_{ess}(\mu) =\Gamma_\phi$ and $\tfrac{{\rm d}\mu}{{\rm d}\theta} > 0$ for almost every $\theta\in[\phi,2\pi-\phi]$, then the Verblunsky coefficients satisfy the L\'{o}pez conditions~\eqref{L1}--\eqref{L2} with $|\alpha|=\sin(\phi/2)$. Combined with Proposition \ref{prop:alpha}, this immediately gives the following corollary.
\begin{corollary}\label{cor:alpha}
Let $\mu$ be such that $\sigma_{ess}(\mu) =\Gamma_\phi$ and $\tfrac{{\rm d}\mu}{{\rm d}\theta} > 0$ for almost every $\theta\in[\phi,2\pi-\phi].$ Then \eqref{eq:generalSSLT1cor1} holds  for $h \in \mathfrak B_{\frac12}$.
\end{corollary}

\begin{remark} (Non-relative) Asymptotics of Toeplitz determinants with measures on a single arc have been studied by Widom \cite{Widom} and Krasovsky \cite{Krasovsky} under smoothness conditions on the measure $\mu$ and for symmetric $h$.  If $h$ is symmetric then $b_j\equiv 0$ and hence $\mathcal B^h=0$. The result then is $$Q_\alpha(h)= \sum_{j=1}^\infty j \mathcal A^h_j \mathcal A^h_j,$$
which is the same  expression as in the Strong Szeg\H{o} Theorem up to the map $\omega$. This result is in agreement with the results by Widom \cite{Widom} and Krasovsky \cite{Krasovsky}.
\end{remark}

\begin{remark} Although all $|\alpha_n| <1$ it is possible to have coefficients such that $ \lim_{n \to \infty} \alpha_n=\alpha\in\bbT$. This case is also contained in our results. The arc collapses to a single point and all the definitions become trivial. In particular, $Q_1(h)=0$.
\end{remark}

\begin{remark}
The above examples are ``living'' on an arc $\Gamma_\phi$ (see~\eqref{eq:arc}) for some $0\le \phi<\pi$. There is nothing special about such symmetric arcs though. One can simply rotate a measure $\mu$ with $\sigma_{ess}(\mu) = \Gamma_\phi$ by a unimodular number $\lambda\in\bbT$ to obtain the analogous results for measures supported on non-symmetric arcs. The corresponding Verblunsky coefficients will have form $\alpha_n = \alpha \lambda^n$ for some $\lambda\in\bbT$, and the L\'{o}pez condition~\eqref{L2} should be modified to $\frac{\alpha_{n+1}}{\alpha_n} \to \lambda$. Throughout this paper we will restrict ourselves to the case of symmetric arcs ($\lambda=1$) for convenience purposes.
\end{remark}

\begin{remark}
In Corollary \ref{cor:alpha} we deal with measures with essential support on an arc. However, in  the more general Theorem \ref{thm2} the measure does not necessarily have such  a support. Indeed, a sequence can have several convergent subsequences. It is possible to construct a sequence of Verblunsky coefficients for which the support is the full circle but we still have \eqref{eq:subslimit} along a subsequence with  $\alpha\neq 0$.
\end{remark}
\begin{remark}
We note that~\eqref{eq:subslimit} is a rather weak condition. For instance, it does not  guarantee the existence of the limit of $n^{-1}\int f(z) K_n(z,z) {\rm d} \mu(z)$ (note that the stronger $\alpha_n \to \alpha$ would). In this sense one can say that second term in the asymptotic behavior of \eqref{eq:ratio} is more robust, or more universal, than  the first term.
\end{remark}\begin{remark}
We recall that a particular motivation for studying ratios of Toeplitz determinants comes from linear statistics for point process defined by \eqref{eq:coulombgas}. From this perspective,  Theorem \ref{thm1} can be regarded as a Central Limit Theorem where the limiting variance is given by $Q_\alpha$. That is,
\begin{multline}
\mathbb E\left[{\rm e} ^{ {\rm i}t \left(X_n(f)-\mathbb E X_n(f)\right)}\right] =\mathbb E[{\rm e} ^{{\rm i} tX_n(f)}]{\rm e}^{-{\rm i} t \mathbb E X_n(f)}
\\
= \frac{D_n({\rm e}^{{\rm i} t f}{\rm d} \mu)}{D_n({\rm d} \mu)} {\rm e}^{-{\rm i} t \int f(z) K_n(z,z) {\rm d} \mu(z)}
=\Psi_n({\rm i} t f,\mu) \to {\rm e}^{-t^2 Q_\alpha(f)},
\end{multline}
as $n \to \infty$.  For a similar discussion in the real line setting we refer to \cite{BDjams}.
\end{remark}
\subsection{Right limits and varying measures}
The previous paragraph dealt with some special limits. For general $\mu$ we  now consider  a subsequence $\{n_j\}$  of $\mathbb N$ such that,
 \begin{equation} \label{eq:rightlimit}
 \lim_{j\to \infty} \alpha_{n_j+k}= \beta_k, \quad k \in \mathbb Z,
 \end{equation}
for some sequence $\{\beta_k\}_{k \in \bbZ}\subset\{ z \mid |z| \leq 1\}$. Indeed, the existence of  such subsequences is guaranteed by a standard compactness argument (we recall that $|\alpha_k|<1$). The sequence $\{\beta_k\}_{k\in \bbZ}$ is called  a \emph{right limit} of the original sequence  $\{\alpha_k\}_{k \in \mathbb N}$.

The next result is that \eqref{eq:rightlimit} implies that we have an analogue of \eqref{eq:Szego} along the subsequence $\{n_j\}$ where the limit is determined by $\{\beta_k\}_{k\in \mathbb N}$.

\begin{theorem}\label{thm:rightlimit}
Let $\mu$ be a measure on $\mathbb T$ and let  $\{\beta_k\}_{k \in \bbZ}$ be a right limit of the Verblunsky coefficients along $\{n_j\}_{j \in \mathbb N}$.
 Then  there exists a function $q:  \mathfrak B _{\frac{1}{2}}\to \mathbb C$, determined solely by the sequence $\{\beta_k\}_{k \in \bbZ}$, such that
\begin{equation}
\lim_{j\to \infty } \Psi_{n_j} (h,\mu)= q(h),
\end{equation}
 for $h \in  \mathfrak B _{\frac{1}{2}}$.  If $h$ is real, then $q(h)$ is positive and $h \mapsto q(h)$ is continuous  with respect to the $\|\cdot \|_{ \mathfrak B _{\frac12}}$-norm.
\end{theorem}

The proof of this theorem will be given in Section \ref{sec:proofrightlimit}.

The remarkable conclusion is that there is an analogue of \eqref{eq:Szego} for every right limit of the sequence of the Verblunsky coefficients.   The function $h \mapsto q(h)$ is explained in Section \ref{sec:proofrightlimit}, but the construction we provide here is not explicit. We leave it as an interesting open question to find a more tangible expression for  $q(h)$ for an arbitrary sequence $\{\beta_k\}_{k \in \bbZ}$. Observe that in case \eqref{eq:subslimit} we have $q={\rm e}^{Q_\alpha}$.

Theorem \ref{thm:rightlimit} should be compared to an analogous result \cite[Th. 2.4]{BDjams} for the real-line setting using
right limits for Jacobi matrices. An important difference is that  \cite[Th. 2.4]{BDjams} holds for a rather restrictive class of functions $h$ (polynomials with sufficiently small sup-norm), whereas Theorem \ref{thm:rightlimit} holds for  $h \in \mathfrak B_{\frac12}$. The reason that we can allow a more general class, is that in the present setting we can achieve some important inequalities  (cf. Lemma \ref{lem:normalfamily})  by  improving  results from \cite{BDadv}.

\begin{remark}
Note that   right limits are invariant under small perturbations. Indeed, if we perturb the sequence of Verblunsky coefficients and consider $\{\alpha_k+ \eps_k\}_{k \in \mathbb N}$ for  a sequence $\{\eps_k\}_{k \in \mathbb N}$ such that $\eps_k \to  0$ as $k \to \infty$, then the right limits do not change.
\end{remark}

Finally, we mention another generalization. The proofs that we present here, work also in case $\mu$ varies with $n$. Let $\{\mu_n\}_{n \in \mathbb N}$ be a sequence of measures and denote the Verblunsky coefficients of $\mu_n$ by  $\{\alpha_{k}^{(n)}\}_{k =0}^{\infty}$. Then our main results are all valid with the appropriate adjustment of notation. For instance, we have the following Theorem.

\begin{theorem}\label{thm:rightlimitvarying}
Let $\{\mu_n\}_{n \in \mathbb N}$ be a sequence of measure on $\mathbb T$ and let   $\{\beta_k\}_{k \in \bbZ}$ be a right limit of the Verblunsky coefficients along $\{n_j\}_{j \in \mathbb N}$, i.e.
$$\lim_{j \to \infty}\alpha_{n_j+k}^{(n_j)} = \beta_k,$$
for $k \in \mathbb Z$. Then with $q$ as in Theorem \ref{thm:rightlimit}
\begin{equation}
\lim_{j\to \infty } \Psi_{n_j} (h,\mu_{n_j})= q(h),
\end{equation}
 for $h \in  \mathfrak B _{\frac{1}{2}}$.
\end{theorem}

In order to avoid cumbersome notation we will not prove this explicitly in this paper and work with fixed measures only. But the generalization to varying measures is straightforward.

\subsection{Weak asymptotics}\label{ssWeak}

From~\eqref{eq:heine}, a Toeplitz determinant $D_n({\rm d} \mu)$ is equal to $\prod_{j=0}^{n-1} \kappa_j(\mu)^{-2}$, where $\kappa_j(\mu)$ is the leading coefficient of the $j$-th orthonormal polynomial associated to $\mu$.
There are numerous papers in the literature on various asymptotics of $\kappa_j$'s. 
In particular, the currently unresolved Nevai's conjecture (\cite[Sect 2.9]{OPUC1}, see also~\cite{MNT2,MNT3}) deals with the asymptotics of the ratios $\frac{\kappa_j(\tilde\mu)}{\kappa_j(\mu)}$. 
A Ces\`{a}ro-type asymptotics of ratios of $\kappa_j$'s was established by Simon~\cite[Thm 9.10.4]{OPUC2}.
In the language of Toeplitz determinants, it says that
if $\frac{1}{n} K_n(z,z) {\rm d} \mu(z)$ has a weak limit $\nu$, i.e.
$$\frac{1}{n} \int h(z) K_n(z,z) {\rm d} \mu(z) \to \int h(z) {\rm d} \nu(z),$$
then 
$$
\lim_{n\to\infty} \left( \frac{D_n(e^h  {\rm d}\mu)}{D_n({\rm d}\mu)}\right)^{1/n} =e^{ \int h(z) {\rm d} \nu(z)},
$$
for any continuous and real-valued $h$.

This can be viewed as a special case of the following result that holds without any condition on $\mu$.
\begin{proposition}\label{prop:weak}
Let $\mu$ be a Borel measure  and  $h:\mathbb T \to \mathbb C$ a continuous function such that $e^h$ is sectorial. That is, there exist $\tau \in \mathbb T$ and $\eps>0$  such that $\Re \tau e^h\geq \eps$. Then we have\begin{equation}
\label{eq:firstTermSzego2}
\lim_{n\to\infty} \left( \frac{D_n(e^h {\rm d}\mu)}{D_n({\rm d}\mu)}\right)^{1/n}  e^{- \frac{1}{n} \int h(z) K_n(z,z) {\rm d} \mu(z)}= 1.
\end{equation}
\end{proposition}

The proof of this proposition will be given in Section \ref{sec:proofweak}.

The proposition is stated for sectorial symbols. Such symbols have the important property that the Toeplitz determinants $D_n(e^h {\rm d} \mu)$ do not vanish (see \cite[Prop 2.17]{BS}). Indeed, if $\lambda $ is an eigenvalue of the finite Toeplitz matrix $ T_n(\tau e^h {\rm d} \mu)$, then it has an eigenvector $\psi$ (normalized to $\|\psi\|=1$) and $$\lambda= (T_n(\tau e^h) \psi,\psi)_{\mathbb C^n}
= \int \tau e^{h(z)} \left|\sum_{j=0}^{n-1} \psi_j z^j \right|^2 {\rm d}\mu(z).$$
By taking real parts at both sides we find  $\Re \lambda>\eps$. Therefore we see that none of the eigenvalues vanishes and thus also the determinant $D_n(\tau e^h {\rm d} \mu)$ is non-zero. Since, $ D_n(\tau e^h {\rm d}\mu)= \tau^n D_n({\rm e}^h {\rm d}\mu)$ we also have that $D_n({\rm e}^h {\rm d} \mu)$ does not vanish as claimed. This property of sectorial symbols will be relevant in our proof.

\begin{remark}
The continuity of $h$ in Proposition \ref{prop:weak} is slightly stronger than we need. The proof that we present here works for any function $h$ such that
\begin{equation} \label{eq:variancehjh} \frac1n \iint |h(z)-h(w)|^2 |K_n(z,w)|^2 {\rm d} \mu(w) {\rm d}\mu(z) \to 0,
\end{equation}
as $n \to \infty$. For continuous function $h$ this holds without any conditions on $\mu$, as we will see. With certain extra conditions on the measure $\mu$ one may allow larger classes of functions. Moreover, the rate of convergence in \eqref{eq:firstTermSzego2} is the same as the rate of convergence in \eqref{eq:variancehjh}.\end{remark}

\subsection{Fisher--Hartwig asymptotics}\label{sec:FH}

We now briefly comment on the particular case of Fisher--Hartwig measures.
In the Fisher--Hartwig setup we consider measure of the form
$${\rm d} \mu_{FH}(\theta)= z^{\sum_{j=0}^m \beta_j } \prod_{j=0}^m |z-z_j|^{2 \alpha_j} g_{z_j,\beta_j} (z) z_j^{-\beta_j} {\rm d} \theta, \qquad z= {\rm e}^{{\rm i} \theta},$$
and $$z_j= {\rm e}^{{\rm i }  \theta_j}, \qquad j=0,\ldots,m, \qquad 0 = \theta_0 < \theta_1 < \ldots < \theta_m<  2 \pi,$$
$$
g_{z_j,\beta_j} = \begin{cases}
{\rm e}^{ {\rm i} \pi \beta_j} , & 0 \leq \arg z < \theta_j\\
{\rm e}^{ -{\rm i} \pi \beta_j} , & \theta_j\leq \arg z < 2 \pi\end{cases}
$$
$$\Re \alpha_j>-\frac12, \qquad \beta_j \in \mathbb C, \qquad j=0,1, \ldots,m.$$
The question of the asymptotic behavior of $D_n ({\rm e}^V {\rm d} \mu_{FH})$ is a classical problem with origins in  the Ising model. In several works in the past two decades  this asymptotics has been computed under various assumptions on the parameters. We mention only \cite{DIKannals,EhrhF} and \cite{DIK} for a survey.

In the setting of the present paper it is of interest to see how the asymptotic behavior depends on $V$.  Under certain conditions on the parameters $\alpha_j$ and $\beta_j$ we deduce the following asymptotic behavior for the relative asymptotics from the asymptotic in, e.g. \cite{DIKannals},
\begin{equation}
\frac{D_n({\rm e}^V {\rm d} \mu_{FH})}{D_n({\rm d} \mu_{FH})}
= {\rm e}^{n V_0 +  \sum_{k=1}^\infty k V_k V_{-k}} \prod_{j=1}^r {\rm e}^{-(\alpha_j+\beta_j) \sum_{k>0} V_k z_j^k-(\alpha_j-\beta_j) \sum_{k>0} V_{-k} z_j^{-k}}
\times 
 (1+o(1)),
\end{equation}
as $n \to \infty$. By comparing this to Corollary \ref{cor:alpha0} our results match with this computation after verifying
\begin{equation*}
\int V(z) K_n(z,z) {\rm d}\mu_{FH}(z)
= n V_0-(\alpha_j+\beta_j) \sum_{k>0} V_k z_j^k-(\alpha_j-\beta_j) \sum_{k>0} V_{-k} z_j^{-k} + o(1).
\end{equation*}
This can be verified since the asymptotic of the orthogonal polynomials, and hence the kernel $K_n(z,z)$,  with respect to a Fisher--Hartwig measure is known, e.g. \cite{DIKannals}.  We leave the details to the reader (it may be of help to take $V$ first to be analytic in an annulus and deform the contour of integration, to avoid having to deal with the different asymptotics for the orthogonal polynomials near the singularities).

\subsection{Overview  of the rest of the paper}
The rest of this paper is organized as follows. In Section 3 we  briefly recall some definitions and notions that we  need. In Section 4 we introduce the CMV matrix corresponding to the measure $\mu$. In particular, we rewrite $\Psi_n(h,\mu)$ as a Fredholm determinant and analyze boundedness and continuity properties that we need. Then in Section 5 we  prove Theorems \ref{thm1} and \ref{thm:rightlimit}, together with Proposition \ref{prop:weak}. Finally, in Section 6 we prove Theorem \ref{thm2} and Proposition \ref{prop:alpha}.

\section{Preliminaries}
We start by setting some notation and recalling some basic definitions that we use. For more background on traces and determinants of operators we refer to \cite{SimonTrace} and for Toeplitz operators to \cite{BS}.
\subsubsection*{Function norms}
If $h$ is a function on $\mathbb T$ then we denote the sup-norm of $h$ by $\|h\|_\infty$.

We recall that we defined the  space $ \mathfrak B _{\frac12}$ as
$$ \mathfrak B_{\frac12}=\left\{ h : \mathbb T \to \mathbb C \ \mid \ \|h\|_{ \mathfrak B _{\frac12}}:= \sum_{k \in \mathbb Z} \sqrt{1+|k|} |h_k|< \infty \right \}.$$
This space is  a unital commutative Banach algebra, also an example of a Beurling algebra. This means in particular that, for $g,h \in  \mathfrak B _{\frac12},$ $$\|gh\|_{ \mathfrak B _{\frac12}}  \leq \|g\|_{ \mathfrak B _{\frac12}} \|h\|_{ \mathfrak B _{\frac12}}.$$ Moreover, for any $h \in  \mathfrak B _{\frac12}$ we also have ${\rm e}^h \in  \mathfrak B _{ \mathfrak B _{\frac12}}$ and
$$ \|{\rm e} ^h \|_{ \mathfrak B _{\frac12}} \leq  {\rm e} ^{\| h\|_{ \mathfrak B _{\frac12}}}.$$
We also note that $$\|h\|_{ \mathfrak B _{\frac{1}{2}}}\geq \sum_{j} |h_j|  \geq \|h\|_\infty,$$
showing  in particular that functions in $ \mathfrak B _{\frac12}$  are  continuous and thus also bounded. Finally, we note that
\begin{equation}\label{eq:sobolevbounds}
\left(\sum_{j=-\infty}^\infty |j| |f_j|^2\right)^{\frac12} \leq \sum_{j=-\infty}^\infty \sqrt {|j|} |f_j|\leq \|f\|_{ \mathfrak B _{\frac{1}{2}}}.
\end{equation}

\subsubsection*{Operator norms}
Let $\mathcal H$ be a separable Hilbert space (we will mostly have $\mathcal H=\ell_2(\mathbb N)$ or $\mathcal H=\ell_2(\mathbb Z)$). Then the singular values $\sigma_j(A)$ of a compact operator $A$ are defined as the positive square roots of the eigenvalues of $A^*A$.

We denote the operator-, trace- and Hilbert-Schmidt norms by
\begin{align}
\|A\|_\infty&= \sup_j \sigma_j(A),\\
\|A\|_1 &= \sum_{j=1}^\infty \sigma_j(A),\\
\|A\|_2 & =\left( \sum_{j=1}^\infty \sigma_j(A)^2\right)^{1/2}.
\end{align}
The following well-known inequalities will be used frequently,
$$\|A B \|_j\leq \|A\|_j \|B\|_\infty, \qquad \|A B \|_j\leq \|A\|_\infty \|B\|_j, \quad j=1,2, \infty,$$
and $$\|AB\|_1\leq \|A\|_2 \|B\|_2.$$
Moreover, if $A$ has rank $r<\infty$, then
\begin{equation}\label{eq:finiterank}
\|A\|_1 \leq r \|A\|_\infty, \qquad \|A\|_2\leq \sqrt r \|A\|_\infty.
\end{equation}
Finally, if $\|A\|_1<\infty $ we can define the trace $\Tr A$ (by extending the trace of finite rank operators) which satisfies
$$|\Tr A| \leq \|A\|_1.$$
Similarly, if $\|A\|_1<\infty $ we can define the determinant $\det (I+ A)$  by extending the determinant for finite rank operators. If in addition $\|A\|_\infty<1$, then we have
$$\det (1+A)=\exp \left(\sum_{j=1}^\infty \frac{(-1)^j}{j} \Tr A^j\right).$$

\subsubsection*{Toeplitz and Hankel operators}
If $a= \sum_j a_j z^j$ then the Toeplitz operator $T(a)$ and Hankel operator $H(a)$ are defined as the semi-infinite matrices
$$(T(a))_{jk}= a_{j-k}, \qquad (H(a))_{j+k-1}, \qquad j,k=1,2, \ldots$$
Then $T(a)$ and $H(a)$ are bounded operators on $\ell_2(\mathbb N)$ and
$$\|T(a)\|_\infty = \|a\|_\infty, \qquad \|H(a)\|_\infty\leq \|a\|_\infty.$$
Moreover,  the Hilbert-Schmidt norm  $H(a)$ is given by
$$\|H(a)\|_2^2= \Tr H(a)^*H(a)= \sum_{j=1}^\infty j |a_j|^2.$$
By combining this with \eqref{eq:sobolevbounds} we see that if $a \in  \mathfrak B _{\frac12}$ then $H(a)$ is a Hilbert-Schmidt operator and $\|H(a)\|_2 \leq \|a\|_{ \mathfrak B _{\frac12}}$.

\section{CMV matrices and a Fredholm determinant}

In this section we recall the definition of CMV matrices and rewrite $\Psi_n$ in \eqref{eq:strongasympterm} as a Fredholm determinant.   We will also determine continuity properties of this determinant.

\subsection{CMV matrices}\label{ssCMV}

The CMV operator  is a natural object in the theory of orthogonal polynomials on the unit circle (see Cantero--Moral--Vel\'{a}zquez paper~\cite{CMV}).
By applying the Gram--Schmidt procedure to the sequence $\{1,z,z^{-1},z^2,z^{-2},\ldots\}$ in $L^2(\mu)$, one obtains a sequence $\{\chi_n\}_{n=0}^\infty$ of Laurent polynomials which is a basis of $L^2(\mu)$. With respect to this basis the operator of multiplication by $z$ has a matrix representation $\calC$ (that is, $\calC_{jk} = \int z \overline{\chi_j(z)} \chi_k(z) {\rm d}\mu$) given by
\begin{equation*}
\calC=\left(
\begin{array}{ccccccc}
 \bar{\alpha}_0  & \bar{\alpha}_1 \rho_0  & \rho_0 \rho_1 & 0 & 0 & 0 & \cdots  \\
 \rho_0 & -\alpha_0 \bar{\alpha}_1  & -\alpha_0 \rho_1 & 0 & 0 & 0 & \cdots  \\
0 & \bar{\alpha}_2 \rho_1 & -\alpha_1 \bar{\alpha}_2  & \bar{\alpha}_3 \rho_2 & \rho_2 \rho_3 & 0 & \cdots  \\
0 & \rho_1 \rho_2  & -{\alpha}_1 \rho_2 & -\alpha_2 \bar{\alpha}_3  & -\alpha_2\rho_3 & 0 & \cdots \\
0 & 0 & 0 & \bar{\alpha}_4 \rho_3 & -\alpha_3  \bar{\alpha}_4 & \bar{\alpha}_5\rho_4 & \cdots   \\
0 & 0 & 0 & \rho_3 \rho_4 & -\alpha_3 \rho_4 & -\alpha_4 \bar{\alpha}_5 & \cdots  \\
\cdots & \cdots & \cdots & \cdots & \cdots & \cdots & \cdots  \\
\end{array}
\right),
\end{equation*}
where 
$\alpha_j$, $j\ge0$, are the same Verblunsky coefficients as in~\eqref{eq:rec},  and $\rho_j:=\sqrt{1-|\alpha_j|^2}$. Note that $\calC$ is a unitary operator on $\ell^2(\mathbb N)$.

The following proposition shows the relation between the ratio of Toeplitz determinants that we want to consider in this paper and the CMV matrices associated to $\mu$.
\begin{proposition} We have
$$
\frac{D_n({\rm e}^h {\rm d} \mu)}{D_n({\rm d} \mu)}=
 \det \left( I +P_n ({\rm e}^{h(\mathcal C)}-I) P_n \right)_{\ell_2(\N)},
$$
where $P_n$ is the projection on the first $n$ coefficients, i.e.
$$ P_n e_j=\begin{cases}
e_j & j=1,2, \ldots, n,\\
0& \text{ otherwise,}
\end{cases}$$
and $Q_n=I-P_n$.
\end{proposition}
\begin{proof}
We recall that the Toeplitz determinant is defined as
$$D_n({\rm e}^h {\rm d} \mu)=\det \left(\int_{\mathbb T} z^j \overline{z}^k {\rm e}^{h(z)} {\rm d} \mu(z)\right)_{j,k=1}^n.$$
Since $|z|=1$ we can also write this as
$$D_n({\rm e}^h {\rm d} \mu)=\det \left(\int_{\mathbb T} z^{j-\lfloor (n-1)/2\rfloor-1}  \overline{z}^{k-\lfloor (n-1)/2 \rfloor-1} {\rm e}^{h(z)} {\rm d} \mu(z)\right)_{j,k=1}^n.$$
where $ {\lfloor (n-1)/2 \rfloor}$ is the largest integer less or equal to $(n-1)/2$.
By taking linear combination of the rows and columns we can write
$$
D_n({\rm e}^h {\rm d} \mu)=c \det \left(\int_{\mathbb T} \chi_{j-1}(z) \overline{\chi_{k-1}(z)} {\rm e}^{h(z)} {\rm d} \mu(z)\right)_{j,k=1}^n,
$$
where $\{\chi_j\}_{j=0}^\infty$ are the orthonormal functions that we used in defining the CMV matrix $\mathcal C$. The constant $c$  depends only on the $\chi_j$ and $n$, but not on $h$, and hence it can  be computed by taking the special case $h=0$, giving
$$D_n({\rm d} \mu)=c.$$
Therefore we have
$$
\frac{D_n({\rm e}^h {\rm d} \mu)}{D_n({\rm d}\mu)}= \det \left(\int_{\mathbb T} \chi_{j-1}(z) \overline{\chi_{k-1}(z)} {\rm e}^{h(z)} {\rm d} \mu(z)\right)_{j,k=1}^n.
$$
Now we use that multiplication by $z$ in the basis functions
$\{\chi_j\}_{j=0}^\infty$ is equivalent to  $\mathcal C$  and hence multiplication by ${\rm e}^{h(z)}$ to ${\rm e}^{h(\mathcal C)}$. Hence the matrix in the determinant on the right-hand side is the $n\times n$ upper left block of ${\rm e}^{h(\mathcal C)}$. This gives the statement.\end{proof}

Note that the integral operator on $L_2 (\mu)$ with kernel $K_n$  is the projection operator onto the span of $\{1,z,z^2,\ldots,z^{n-1}\}$, while $P_n$ is the projection onto the span of $\{\chi_0,\ldots,\chi_{n-1}\}$. So these projections are related to each other via the conjugation by $z^{\lfloor (n-1)/2\rfloor}$. Since $h(\calC)$, viewed as the operator of multiplication by $h(z)$ in $L^2(\mu)$, commutes with multiplication by $z$, we obtain  the following equality of traces:
$$\int h(z) K_n(z,z) {\rm d} \mu(z)= \Tr P_n h(\mathcal C)P_n.$$
For this reason we can write $\Psi_n(h,\mu)=\Psi_n(h,\mathcal C)$ as
\begin{equation}
\Psi_n(h,\mathcal C)
= \det\left (I + P_n({\rm e}^{ h(\mathcal C) }-I)P_n \right) {\rm e}^{-\Tr P_n h(\mathcal C) P_n}.
\end{equation}
Here $\mathcal C$ is the CMV-matrix corresponding to $\mu$ and $h \in  \mathfrak B _{\frac12}$. In the proofs we will mostly write $\Psi_n(h,\mathcal C)$ instead of $\Psi_n(h,\mu)$.

\begin{remark}
In view of the remark just below Proposition \ref{prop:weak} we also mention that
$$\iint |h(z)-h(w)|^2 |K_n(z,w)|^2 {\rm d} \mu(z) {\rm d} \mu(z)= \| [P_n,h(\mathcal C)]\|_2^2.$$
The commutator on the right-hand side, and its Hilbert-Schmidt norm, will appear frequently in the coming proofs.
\end{remark}

\subsection{The auxiliary function $\Phi_n$}
We will make extensive use of an auxiliary function $\Phi_n$. Before defining this function, we first note that it follows from \eqref{eq:heine} that
$$\det (I+ P_n({\rm e}^{th(\mathcal C)}-I)P_n)= \frac{D_n({\rm e}^{th} {\rm d} \mu) }{D_n({\rm d} \mu) }>0,$$
if $t \in \mathbb R$ and $h$ real-valued. Now define the auxiliary function $ \Phi_n(t,h,\mathcal C)$ by
$$
 \Phi_n(t,h,\mathcal C)
= \log \det\left (I + P_n({\rm e}^{t h(\mathcal C) }-I)P_n\right)- t \Tr P_n h(\mathcal C) P_n.
$$
Note that
\begin{equation} \label{eq:relationPsiandPhi}
\Psi_n(h,\mathcal C)= \exp\left(\Phi_n(1,h,\mathcal C)\right).\end{equation} For this reason, we are mainly interested in the value of $\Phi_n(t,h,\mathcal C)$ at $t=1$, but we will also use properties of the function near $t=0$. In fact, we will consider all $t$ in an $\eps>0$ neighbourhood of the interval $[0,1]$, i.e. $t \in \cup_{x\in [0,1]} B_{x,\eps}$ with $B_{x,\eps}= \{z \in \C \mid |z-x| \leq \eps\}$.   If $h$ is real, we can indeed choose $\eps$ sufficiently small such that $\Phi_n(t,h,\mathcal C)$ is  well-defined and analytic in that set.

\begin{lemma}\label{lem:normalfamily}
Let  $h \in \mathbb L_\infty(\mathbb T)$  be real-valued. Let $\eps>0$  be any  sufficiently small number such that
\begin{equation}\label{eq:defrhoh}
\rho_h:=\essinf_{\overset{z \in \mathbb T}{t\in \bigcup_{x \in [0,1]} B_{x,\eps}}} \Re {\rm e}^{t h(z)}>0.
\end{equation}
Then $\Phi_n(t,h,\mathcal C)$ is a well-defined function that is analytic in $t \in \cup_{x\in [0,1]} B_{x,\eps}$ and
\begin{equation} \label{eq:derivativephiincommutators}
\frac{{\rm d} }{{\rm d} t} \Phi_n(t,h,\mathcal C)= -\Tr \left(I+P_n({\rm e} ^{t h(\mathcal C)}  -I)P_n\right)^{-1} P_n [{\rm e} ^{t h(\mathcal C)},P_n] [h(\mathcal C),P_n]P_n.
\end{equation}
Moreover,
\begin{equation} \label{eq:bounddervativephiincommutatorsbounded}
\frac1n  |\Phi_n(t,h,\mathcal C)| \leq \frac{8(1+ \eps)}{\min(\rho_h,1)} ({\rm e}^{(1+\eps)\|h\|_{
\infty}}-1) \|h\|_{\infty}.
 \end{equation}
for $t\in \bigcup_{x \in [0,1]} B_{x,\eps}$ and $n \in \mathbb N$.

If $h \in \mathfrak B_{\frac12}$ is real-valued, then
\begin{equation} \label{eq:bounddervativephiincommutators}
 |\Phi_n(t,h,\mathcal C)| \leq \frac{16(1+ \eps)}{\min(\rho_h,1)} {({\rm e}^{(1+\eps)\|h\|_{ \mathfrak B _{\tiny{\frac12}}}}-1) \|h\|_{ \mathfrak B _{\frac12}}}.
 \end{equation}
for $t\in \bigcup_{x \in [0,1]} B_{x,\eps}$ and $n \in \mathbb N$.
\end{lemma}
\begin{proof}
From the discussion preceding the lemma we know that  $\Phi_n(t,h,\mathcal C)$ is well-defined and analytic in a neighborhood of $[0,1]$. It remains to check that this neighborhood can be taken to be $\cup_{x \in [0,1]} B_{x,\eps}$ with $\eps$ as indicated in the statement. Note that by taking the derivative with respect to $t$ and using the identity
$$\frac{{\rm d} }{{\rm d} t} \log \det \left(I+A(t)\right)= \Tr (I+A(t))^{-1}{A'(t)},$$
we obtain
\begin{equation}\label{eq:derivativephi}
\frac{{\rm d} }{{\rm d} t} \Phi_n(t,h)= \Tr \left(I+P_n({\rm e} ^{t h(\mathcal C)} -I)P_n\right)^{-1}P_n{\rm e} ^{t h(\mathcal C)} h(\mathcal C) P_n- \Tr P_n h(\mathcal C)P_n.
\end{equation}
We will show that
$$\left(I+P_n({\rm e} ^{t h(\mathcal C)}  -I)P_n\right)^{-1}$$
exists and is analytic  in $t \in \cup_{x\in [0,1]} B_{x,\eps}$. To this end, we recall the well-known fact  that if $A$ is an operator on a  Hilbert space $\mathcal H$ for which there exists an $r>0$ such that  $\Re (A \psi,\psi)>r(\psi,\psi)$ for $\psi \in \mathcal H$, then $A^{-1}$ exists and $\|A^{-1}\| \leq 1/r$. Indeed, this follows easily from
$$0 \leq \|(A-r)\psi\|^2 = \|A \psi\|^2-2 r \Re (A\psi,\psi) + r^2 \| \psi\|^2,$$
which together with $\Re (A \psi,\psi)>r(\psi,\psi)$ gives $\|A\psi\|^2 \geq r^2 \|\psi\|^2$. From the assumption  in the lemma and using the fact that $h(\mathcal C)$ is unitarily equivalent to multiplication by $h(z)$ in $\mathbb L_2(\mu)$, it is straight-forward to check that
\begin{multline} \Re \left(\left(I+P_n({\rm e} ^{t h(\mathcal C)}  -I)P_n\right) \psi,\psi\right)= \Re (Q_n \psi, \psi)+   \Re \left(P_n{\rm e} ^{t h(\mathcal C)}  P_n\psi,\psi\right)\\=  \Re (Q_n \psi, Q_n \psi)+\Re \left({\rm e} ^{t h(\mathcal C)}  P_n\psi,P_n \psi\right)\geq \min (\rho_h,1) \|\psi\|^2,
\end{multline}
and hence
\begin{equation}\label{eq:boundonresolvent}
\left\| \left(I+P_n({\rm e} ^{t h(\mathcal C)} -I)P_n\right)^{-1}\right\|_\infty \leq 1/\min(\rho_h,1).
\end{equation}
In particular, the function $\Phi_n(t,h,\mathcal C)$ is indeed well-defined and analytic in $\cup_{x\in [0,1]} B_{x,\eps}$.

To prove \eqref{eq:bounddervativephiincommutators}, we start with \eqref{eq:derivativephi} and bring both terms together, giving
\begin{multline}\label{eq:310}
\frac{{\rm d} }{{\rm d} t} \Phi_n(t,h) = \Tr \left(I+P_n({\rm e} ^{t h(\mathcal C)} -I)P_n\right)^{-1}P_n{\rm e} ^{t h(\mathcal C)} h(\mathcal C) P_n\\
- \left(I+P_n({\rm e} ^{t h(\mathcal C)} -I)P_n\right)^{-1}  \left(I+P_n({\rm e}^{t h(\mathcal C)}-I) P_n\right)P_n h(\mathcal C)P_n\\
= \Tr \left(I+P_n({\rm e} ^{t h(\mathcal C)}  -I)P_n\right)^{-1}\left(P_n{\rm e} ^{t h(\mathcal C)} h(\mathcal C)P_n-P_n{\rm e} ^{t h(\mathcal C)} P_n h(\mathcal C)P_n\right)
\\= \Tr \left(I+P_n({\rm e} ^{t h(\mathcal C)}  -I)P_n\right)^{-1}\left(P_n{\rm e} ^{t h(\mathcal C)} [h(\mathcal C),P_n]P_n\right),
\end{multline}
where we used $P_n^2=P_n$. For the same reason, we have  $P_n[h(\mathcal C),P_n]P_n=0$ and we thus obtain \eqref{eq:derivativephiincommutators}.

We then note that  for any three operators $X$, $Y$ and $Z$ we have
\begin{equation} \label{eq:XYZ}
|\Tr XYZ| \leq \|X\|_\infty \|YZ\|_1 \leq  \|X\|_\infty \|Y\|_2\|Z\|_2,
\end{equation} we are left with estimating $\|[{\rm e} ^{t h(\mathcal C)},P_n]\|_2$ and $ \|[h(\mathcal C),P_n]\|_2$.

It is clear that the ranks of $[h(\mathcal C),P_n]$ and $[ {\rm e}^{t h(\mathcal C)}, P_n]$ are both at most $2n$.  By \eqref{eq:finiterank} we find
 \begin{equation}\label{eq:boundhCPn0} \|[ h(\mathcal C), P_n]\|_2 \leq \sqrt{2n}  \|[ h(\mathcal C), P_n]\|_\infty \leq 2 \sqrt{2n}  \|h\|_{ \infty},
 \end{equation}
 and
 \begin{equation} \label{eq:boundjj} \|[ {\rm e}^{t h(\mathcal C)}, P_n]\|_2= \|[ {\rm e}^{t h(\mathcal C)}-I, P_n]\|_2 \leq \sqrt {2n} \|[ {\rm e}^{t h(\mathcal C)}-I, P_n]\|_\infty
 \leq 2 \sqrt{2n} ( {\rm e}^{(1+ \eps) \|h\|_{ \infty }}-1).
 \end{equation}
 By combining \eqref{eq:derivativephiincommutators}  with \eqref{eq:boundonresolvent}, \eqref{eq:boundhCPn0} and \eqref{eq:boundjj} we therefore find
 $$\frac1n\left|\frac{{\rm d} }{{\rm d} t} \Phi_n(t,h,\mathcal C) \right|\leq\frac{8}{\min(\rho_h,1)} ( {\rm e}^{(1+ \eps) \|h\|_{ \infty}}-1) \|h\|_{ \infty}$$
Since $\Phi_n(0,h,\mathcal C)=0$ we see that \eqref{eq:bounddervativephiincommutatorsbounded} now follows by integrating the latter inequality over $t$.

Now suppose that we have in addition $h \in \mathfrak B_{\frac12}$. Then we write
$$[ h(\mathcal C), P_n]= \sum_{k \in \mathbb Z} h_k[ \mathcal C^k,P_n].$$
The main point of the proof of \eqref{eq:bounddervativephiincommutators} is that  due to band structure of $\mathcal C$ (and $\mathcal C^{-1}= \mathcal C^*$) we have
 the rank of $[\mathcal C^k,P_n]$ is at most $4|k|$. Hence
 \begin{align}\label{eq:boundhCPn} \|[ h(\mathcal C), P_n]\|_2\leq  2 \sum_{k \in \mathbb Z} \sqrt{|k|} |h_k|\|[\calC^k,P_n]\|_\infty \leq 4 \|h\|_{ \mathfrak B _{\frac12}}.
 \end{align}
 For the same reason we have
 $$\|[ {\rm e}^{t h(\mathcal C)}, P_n]\|_2= \|[ {\rm e}^{t h(\mathcal C)}-I, P_n]\|_2  \leq 4 \| {\rm e}^{t h}-1\|_{ \mathfrak B _{\frac12}}.$$
 Since $\|f g\|_{ \mathfrak B _{\frac12}} \leq \|f\|_{ \mathfrak B _{\frac12}} \|g\|_{ \mathfrak B _{\frac12}}$  and by a Taylor expansion we then find
 \begin{equation}\label{eq:boundexphCPn}
 \|[ {\rm e}^{t h(\mathcal C)}, P_n]\|_2\leq 4( {\rm e}^{|t| \|h\|_{ \mathfrak B _{\frac12}}}-1)\leq 4( {\rm e}^{(1+ \eps) \|h\|_{ \mathfrak B _{\frac12}}}-1).
 \end{equation}
 By combining \eqref{eq:derivativephiincommutators}  with \eqref{eq:boundonresolvent}, \eqref{eq:boundhCPn} and \eqref{eq:boundexphCPn} we therefore find
 $$\left|\frac{{\rm d} }{{\rm d} t} \Phi_n(t,h,\mathcal C) \right|\leq\frac{16}{\min(\rho_h,1)} ( {\rm e}^{(1+ \eps) \|h\|_{ \mathfrak B _{\frac12}}}-1) \|h\|_{ \mathfrak B _{\frac12}}$$
and hence \eqref{eq:bounddervativephiincommutators} now follows after integrating over $t$ again. This finishes the proof.
\end{proof}

We will also need that the function $\Phi_n$ is continuous in $h$ with respect to the $\|\cdot\|_{ \mathfrak B _{\frac12}}$ norm in the following sense.

\begin{lemma}\label{lem:bound1}
Let $h_1  \in  \mathbb L_\infty(\mathbb T)$ be real-valued. Then there exists $\eps>0$ and $r>0$ such that for all real-valued $h_2$ with $\|h_1-h_2\|_{\infty}<r$ we have that $ \Phi_n(t,h_2,\mathcal C)$ is also well-defined for   $t \in \cup_{x \in [0,1]} B_{x,\eps}$ and
\begin{equation}
\label{eq:equicontinuitybound}
\frac1n |\Phi_n(t,h_1,\mathcal C)-\Phi_n(t,h_2,\mathcal C)|= \mathcal O(\|h_1-h_2\|_{ \infty}),
\end{equation}
as $\|h_1-h_2\|_\infty \to 0$, where the constant is uniform for $t \in \cup_{x \in [0,1]} B_{x,\eps}$ and independent of $n$ and $\calC$.

Similarly, for real-valued $h_1\in \mathfrak B_{\frac12}$, we have  that for all real-valued $h_2$ with $\|h_1-h_2\|_{\mathfrak B_{\frac12}}<r$ we have that  $ \Phi_n(t,h_2,\mathcal C)$ is  also well-defined  for $t \in \cup_{x \in [0,1]} B_{x,\eps}$ and
\begin{equation}
\label{eq:equicontinuitysobolev}
|\Phi_n(t,h_1,\mathcal C)-\Phi_n(t,h_2,\mathcal C)|= \mathcal O(\|h_1-h_2\|_{ \mathfrak B _{\frac12}}),
\end{equation}
as $\|h_1-h_2\|_{\mathfrak B_{\frac12}}\to 0$, where the constant is uniform for $t \in \cup_{x \in [0,1]} B_{x,\eps}$ and independent of $n$ and $\calC$.
\end{lemma}
\begin{proof}
It is elementary to show that we can choose $r,\eps>0$ such that
\begin{equation}\label{eq:tilderho}
\tilde \rho_{h_1} = \inf_{\|h_2-h_1\|\leq r} \rho_{h_2} >0,
\end{equation}
where $\rho_h$ is as in \eqref{eq:defrhoh}. Hence
$\Phi_n(t,h,\mathcal C)$ is indeed well-defined and we have the bound \eqref{eq:bounddervativephiincommutators} for $\Phi_n(t,h_2,\mathcal C)$ with $\rho_h$ replaced by $\tilde \rho_{h_1}$. Moreover, by   \eqref{eq:derivativephiincommutators} we find
\begin{multline*}
\frac{{\rm d} }{{\rm d} t} \Phi_n(t,h_1,\mathcal C)-\frac{{\rm d} }{{\rm d} t} \Phi_n(t,h_2,\mathcal C)= \Tr R(h_2)\left(P_n[{\rm e} ^{t h_2(\mathcal C)},P_n] [h_2(\mathcal C),P_n]P_n\right)\\-\Tr R(h_1)\left(P_n[{\rm e} ^{t h_1(\mathcal C)},P_n] [h_1(\mathcal C),P_n]P_n\right),
\end{multline*}
where
$$R(h)=\left(I+P_n({\rm e} ^{t h(\mathcal C)}  -I)P_n\right)^{-1}.$$
We rewrite this as
\begin{multline} \label{eq:derivativephiincommutatorsdiff}
\frac{{\rm d} }{{\rm d} t} \Phi_n(t,h_1,\mathcal C)-\frac{{\rm d} }{{\rm d} t} \Phi_n(t,h_2,\mathcal C)\\
=  \Tr (R(h_2)-R(h_1))\left(P_n[{\rm e} ^{t h_1(\mathcal C)},P_n] [h_1(\mathcal C),P_n]P_n\right)\\
-\Tr R(h_2)\left(P_n([{\rm e} ^{t h_1(\mathcal C)},P_n] [h_1(\mathcal C),P_n]-[{\rm e} ^{t h_2(\mathcal C)},P_n] [h_2(\mathcal C),P_n]P_n\right),
\end{multline}
and deal with  the two terms at the right-hand side separately, starting with the first.

By the resolvent identity we have
$$R(h_1)-R(h_2)
=R(h_1) \left(P_n({\rm e} ^{t h_2(\mathcal C)} - {\rm e} ^{t h_1(\mathcal C)})P_n\right) R(h_2).
$$
Combining this with
\begin{equation}
\|{\rm e} ^{t h_1(\mathcal C)}- {\rm e} ^{t h_2(\mathcal C)}\|_\infty  \leq\|{\rm e} ^{t h_1(\mathcal C)}\|_\infty \|I-{\rm e} ^{t (h_2(\mathcal C)-h_1(\mathcal C))}\|_\infty  \leq{\rm e}^{(1+\eps)\|h_1\|_\infty} ({\rm e}^{(1+\eps)\|h_1-h_2\|_\infty}-1),
\end{equation}
and the fact that \eqref{eq:tilderho} implies $\|R(h_j)\|_\infty \leq (\min(\tilde \rho_{h_1},1))^{-1}$, we therefore find
$$
\|R(h_1)-R(h_2)\|_\infty
=\frac{1}{\min(\tilde \rho_{h_1}^{2},1)}
{\rm e} ^{(1+\eps)\|h_1\|_\infty}  ({\rm e}^{(1+\eps)\|h_1-h_2\|_\infty}-1) .$$
This implies that
\begin{multline}\label{eq:deffphipart1}
\left|\Tr (R(h_1)-R(h_2))\left(P_n[{\rm e} ^{t h_1(\mathcal C)},P_n] [h_1(\mathcal C),P_n]P_n\right)\right|\\
\leq \|R(h_1)-R(h_2)\|_\infty\| [{\rm e} ^{t h_1(\mathcal C)},P_n]\|_2\| [h_1(\mathcal C),P_n]\|_2 \leq c_1\|h_1-h_2\|_\infty,
\end{multline}
where $c_1$ is a constant that only depends on $\eps,r$ and $\|h_1\|_{ \infty}$ and no other parameters (in particular not on $n$).

Since (cf. \eqref{eq:boundhCPn0} and \eqref{eq:boundjj})
$$\| [h_1(\mathcal C),P_n]-[h_2(\mathcal C),P_n]\|_2 \leq 2 \sqrt{2n} \|h_1-h_2\|_{ \infty}$$
$$\| [{\rm e}^{h_1(\mathcal C)},P_n]-[{\rm e}^{t h_2(\mathcal C)},P_n]\|_2 \leq 2 \sqrt{2n} \| {\rm e} ^{th_1}-{\rm e}^{th_2}\|_{ \infty}$$
we have
\begin{multline} \label{eq:deffphipart2}
\left|\Tr R(h_2)\left(P_n([{\rm e} ^{t h_1(\mathcal C)},P_n] [h_1(\mathcal C),P_n]-[{\rm e} ^{t h_2(\mathcal C)},P_n] [h_2(\mathcal C),P_n]P_n\right)\right|\\
\leq \|R(h_2)\|_\infty \left\| [{\rm e} ^{t h_1(\mathcal C)},P_n] [h_1(\mathcal C),P_n]-[{\rm e} ^{t h_2(\mathcal C)},P_n] [h_2(\mathcal C),P_n]\right\|_1\\
\leq  \frac{1}{\min(\tilde \rho_{h_1},1)} \left\|( [{\rm e} ^{t h_1(\mathcal C)},P_n]-[{\rm e} ^{t h_2(\mathcal C)},P_n] )\right\|_2\left\| [h_1(\mathcal C),P_n]\right\|_2\\+\frac{1}{\min(\tilde \rho_{h_1},1)}\left\|[{\rm e} ^{t h_2(\mathcal C)},P_n]\right\|_2 \left\| [h_1(\mathcal C),P_n]-[h_2(\mathcal C),P_n]\right\|_2
\leq n c_2 \|h_1-h_2\|_{ \infty},
\end{multline}
where $c_2$ is a constant that depends on $\eps,r$ and $\|h\|_{ \infty}$ but no other parameters (and in particular not on $n$). By substituting \eqref{eq:deffphipart1} and \eqref{eq:deffphipart2} into \eqref{eq:derivativephiincommutatorsdiff} and integrating over $t$ we obtain \eqref{eq:equicontinuitybound}.

For $h_1, h_2 \in \mathfrak B_{\frac12}$ we recall   \eqref{eq:boundhCPn} and \eqref{eq:boundexphCPn} giving
$$\| [h_1(\mathcal C),P_n]-[h_2(\mathcal C),P_n]\|_2 \leq 4 \|h_1-h_2\|_{ \mathfrak B _{\frac12}}$$
$$\| [{\rm e}^{h_1(\mathcal C)},P_n]-[{\rm e}^{t h_2(\mathcal C)},P_n]\|_2 \leq 4 \| {\rm e} ^{th_1}-{\rm e}^{th_2}\|_{ \mathfrak B _{\frac12}}$$
and argue similar as above to obtain \eqref{eq:equicontinuitysobolev}.
\end{proof}
\subsection{Series expansion of $\Phi_n$ around the origin}
One of the main ingredients in the proof are the coefficients in the expansion of $\Phi_n$  around $t=0$.
\begin{lemma} Let $h\in \mathfrak{B}_{\frac12}$ be real-valued.  Then
the function $\Phi_n(t,h,\mathcal C)$ has the series
\begin{equation}\label{eq:cumulseries}
\Phi_n(t, h, \mathcal C)= \sum_{m=1}^\infty t^{m+1}E_m^{(n)}(h(\mathcal C)),
\end{equation}
where
\begin{equation}
\label{eq:cumul}
E_m^{(n)}(h(\mathcal C))=
\frac{1}{m+1}\sum_{j=1}^m {(-1)^{j-1}}\sum_{l_1+ \cdots + l_j=m, l_i \geq 1} \frac{\Tr P_n h(\mathcal C)^{l_1}P_n h(\mathcal C)^{l_2} \cdots P_n h(\mathcal C)^{l_j} [h(\mathcal C),P_n]}{l_1! \ldots l_j!}.
\end{equation}
The series converges for $|t|\leq \frac{1}{{\rm e} \|h\|_\infty}$.
\end{lemma}
\begin{proof}
We start by recalling \eqref{eq:310} giving
$$\frac{{\rm d} }{{\rm d} t} \Phi_n(t,h,\mathcal C)= \Tr \left(I+  P_n({\rm e}^{t h(\mathcal C)}-I)P_n\right)^{-1} P_n{\rm e}^{t h(\mathcal C)} [h(\mathcal C), P_n] P_n.$$
Since $P_n  [h(\mathcal C), P_n] P_n=0$, we can also write this as
$$\frac{{\rm d} }{{\rm d} t} \Phi_n(t,h,\mathcal C)= \Tr \left(I+  P_n({\rm e}^{t h(\mathcal C)}-I)P_n\right)^{-1 } P_n( {\rm e}^{t h(\mathcal C)}-I)  [h(\mathcal C), P_n] P_n.$$
Now we compute the inverse by a Neumann-series, rearrange the order of summation and write
\begin{multline} \label{eq:dphindt}
\frac{{\rm d} }{{\rm d} t} \Phi_n(t,h,\mathcal C)=\Tr \sum_{j=0}^\infty (-1)^j\left(  P_n({\rm e}^{t h(\mathcal C)}-I)P_n\right)^{j} P_n ( {\rm e}^{t h(\mathcal C)}-I)
 [h(\mathcal C), P_n] P_n\\
 =\Tr \sum_{j=1}^\infty (-1)^{j-1}\left(  P_n({\rm e}^{t h(\mathcal C)}-I)\right)^{j} [h(\mathcal C), P_n] P_n\\
  =\Tr \sum_{j=1}^\infty (-1)^{j-1}\sum_{l_1,\ldots,l_j=1}^\infty \frac{t^{l_1+ \cdots + l_j} P_n h(\mathcal C)^{l_1} \cdots P_n h(\mathcal C)^{l_j}}{l_1! \cdots l_j!} [h(\mathcal C), P_n] P_n\\
  =\Tr \sum_{j=1}^\infty (-1)^{j-1}\sum_{m=j}^\infty t^m \sum_{\overset{l_1+\ldots+l_j=m}{l_i \geq 1} } \frac{P_n h(\mathcal C)^{l_1} \cdots P_n h(\mathcal C)^{l_j}}{l_1! \cdots l_j!} [h(\mathcal C), P_n] P_n\\
   =\sum_{m=1}^\infty t^m \sum_{j=1}^m (-1)^{j-1} \sum_{\overset{l_1+\ldots+l_j=m}{l_i \geq 1} } \Tr \frac{P_n h(\mathcal C)^{l_1} \cdots P_n h(\mathcal C)^{l_j}}{l_1! \cdots l_j!} [h(\mathcal C), P_n] P_n
 \end{multline}
 By integrating the later expression over $t$ and using $\Phi_n(0,h,\mathcal C)=0$  we obtain the statement. However, it is not clear whether the last expression is well-defined since the series (over $m$) may diverge, so it remains to estimate the radius of convergence.

 Since $P_n^2=P_n$ and $P_n[h(\mathcal C),P_n]P_n=0$, we have
 \begin{equation*}
  \Tr {P_n h(\mathcal C)^{l_1} \cdots P_n h(\mathcal C)^{l_j}} [h(\mathcal C), P_n] P_n=
   \Tr {P_n h(\mathcal C)^{l_1} \cdots P_n h(\mathcal C)^{l_{j-1}}}[ P_n,h(\mathcal C)^{l_j}] [ h(\mathcal C), P_n]P_n.
 \end{equation*}
 After writing $[P_n,h(\mathcal C)^{l_j}]= \sum_{k=0}^{l_j-1} h(\mathcal C)^k [P_n,h(\mathcal C)] h(\mathcal C)^{l_j-1-k}$ and estimating the trace as in the proof of Lemma \ref{lem:normalfamily} we find
\begin{equation*}
\left|  \Tr {P_n h(\mathcal C)^{l_1} \cdots P_n h(\mathcal C)^{l_j}} [h(\mathcal C), P_n] P_n\right| \leq m \|h(\mathcal C)\|_\infty^{m-1} \|[h(\mathcal C),P_n\|_2^2 \leq 16 m  \|h\|_\infty^{m-1} \|h\|_{ \mathfrak B _{\frac12}}^2.
\end{equation*}
where we also used \eqref{eq:boundhCPn} in the last step.
This means that
\begin{multline}\label{eq:estDm}
\left| \sum_{j=1}^m (-1)^{j-1} \sum_{\overset{l_1+\ldots+l_j=m}{l_i \geq 1} } \Tr \frac{P_n h(\mathcal C)^{l_1} \cdots P_n h(\mathcal C)^{l_j}}{l_1! \cdots l_j!} [h(\mathcal C), P_n] P_n\right|
\\
\leq  16 m \|h\|_\infty^{m-1} \|h\|_{ \mathfrak B _{\frac12}}^2 \sum_{j=1}^m  \sum_{\overset{l_1+\ldots+l_j=m}{l_i \geq 1} }\frac{1}{l_1! \cdots l_j!}
 \end{multline}
 and since
 \begin{equation}\label{combest}
 \sum_{j=1}^m  \sum_{\overset{l_1+\ldots+l_j=m}{l_i \geq 1} }\frac{1}{l_1! \cdots l_j!}   \leq \sum_{j=1}^m  \sum_{l_1+\ldots+l_j=m}\frac{1}{l_1! \cdots l_j!}= \sum_{j=1}^m \frac{(1+1+ \ldots+1)^m}{m!}= \frac{m^{m}}{m!} \leq \frac{m^{1/2}}{\sqrt{2\pi} } {\rm e}^{m}.
 \end{equation}
 Combining \eqref{combest} with \eqref{eq:estDm} we see that
 \begin{equation}\label{eq:boundcmn}
 |E_m^{(n)}(h)|=\| h\|_{ \mathfrak B _{\frac12}}^2 \frac{ 16 m^{3/2}{\rm e} }{\sqrt{2\pi} } (\|h\|_\infty{\rm e})^{m-1} {\rm e}.
 \end{equation}
 Hence the series in \eqref{eq:cumulseries} is indeed convergent for $|t| \leq \frac{1}{{\rm e}  \|h\|_\infty}$.
\end{proof}

Before we come to the main argument in the proof, we first note that in the last proof we used an inequality bounding each coefficient in the expansion. This bound also allows us to use a cut-off of the expansion which will be useful for technical reasons. Indeed, if we define
$$
\Phi_{n,N}(t,h,\mathcal C)= \sum_{m=1}^N t^{m+1}E_m^{(n)}(h(\mathcal C)),
$$
then the difference between $\Phi_n$ and $\Phi_{n,N}$ can be estimated as in the following lemma.
\begin{lemma} \label{lem:bound2} For $h \in \mathfrak B_{\frac12}$, we have
\begin{equation}\label{eq:laatste}
\left|\Phi_n(t,h,\mathcal C)-\Phi_{n,N}(t,h,\mathcal C)\right| \leq \|h\|_{ \mathfrak B _{\frac12}}^2\frac{16 {\rm e} }{\sqrt{2\pi}} \sum_{m=N+1}^\infty m^{3/2} 2^{-m}
\end{equation}
for $|t| \leq 1/(2 {\rm e} \|h\|_\infty)$.
\end{lemma}
\begin{proof} This is a direct consequence of the bound \eqref{eq:boundcmn}.
\end{proof}
Note that the right-hand side of \eqref{eq:laatste} is independent of $n$.
\section{Proofs of Proposition \ref{prop:weak} and Theorems \ref{thm2} and \ref{thm:rightlimit} }
\subsection{Proof of Proposition \ref{prop:weak}}\label{sec:proofweak}
\begin{proof}[Proof of Proposition \ref{prop:weak}]
First assume that $h$ be real-valued and continuous. We deal with sectorial symbols later. Note that
$$\Psi_n(h,\mathcal C)^{1/n} = \exp \left(\frac1n \Phi_n(1,h,\mathcal C)\right).$$
Then from \eqref{eq:equicontinuitybound} we find
$$\Psi_n(h,\mathcal C)^{1/n} = \exp \left(\frac1n \Phi_n(1,\tilde h,\mathcal C )+ \mathcal O(\|\tilde h-h\|_\infty)\right), $$
for  $\tilde h$ sufficiently close to $h$.  Since Laurent polynomials are dense in the space of continuous functions we can for every $\eps>0$  find a Laurent polynomial $\tilde h$ such that
$$\frac1n \Phi_n(1,\tilde h,\mathcal C )- \eps <\frac1n\Phi_n(1,h,\mathcal C) \leq   \frac1n \Phi_n(1,\tilde h,\mathcal C )+ \eps.$$
Note that by \eqref{eq:bounddervativephiincommutators} and the fact that every Laurent polynomial has finite $\mathfrak B_{\frac12}$-norm, we have $$\lim_{n\to \infty} \frac{1}{n} \Phi_n(1,\tilde h,\mathcal C)=0. $$ Therefore,
$$ -\eps < \liminf_{n\to \infty} \frac1n\Phi_n(1,h,\mathcal C)   \leq \limsup_{n\to \infty} \frac1n\Phi_n(1,h,\mathcal C)     <\eps.$$
Hence the statement follows by taking $\eps \downarrow 0$.

To deal with the complex-valued case we use a normal family argument.  We start by  defining
$$h_z= \Re h + z \Im h,$$
so that $h=h_{\rm i}$. From the proof for the real-valued case we learn that for $z \in \mathbb R$ we have
\begin{equation}\label{eq:hzreal0}
\Psi_{n}(h_z,\mathcal C)^{1/n}\to 1.
\end{equation}
 It remains to prove that this also holds for $z= {\rm i}$.

We claim that for each $n \in \mathbb N$ the function $z \to \Psi_{n}(h_z,\mathcal C)^{1/n}$ is  well-defined and analytic in a disk  $|z|< 1+\delta$ for some sufficiently small $\delta>0$ that is independent of $n$.  Since $\Psi_{n}(h_z, \mathcal C)$ is clearly well-defined and analytic, it remains to show that it does not vanish in the disk with radius $1+ \delta$ so that we can take the $n$-th root. From the definition of $\Psi_{n}$ it is clear that it suffices  to show that $D_{n}( {\rm e}^{h_z} {\rm d} \mu)$ does not vanish. We do this by invoking the fact that Toeplitz determinants for sectorial symbols never vanish (see the discussion directly below Proposition \ref{prop:weak}).  For continuous functions $g$ it holds that ${\rm e}^g$ is sectorial if and only if
$$\max_{w \in \mathbb T} \Im g(w)- \min_{w \in \mathbb T} \Im g(w)< \pi.$$
Now,
$$\max_{w \in \mathbb T} \Im h_z(w)- \min_{w \in \mathbb T} \Im h_z(w)= |\Im z| \left(\max_{w \in \mathbb T} \Im h(w)- \min_{w \in \mathbb T} \Im h(w)\right),$$
and since ${\rm e}^h$ is assumed to be sectorial, it is now easy to see that we can choose a $\delta>0$ such that ${\rm e}^{h_z}$ is also sectorial for each $|z|< 1+\delta.$  This proves the claim that $D_{n}({\rm e}^{h_z}{\rm d} \mu)$, and thus also $\Psi_{n}(h_z, \mathcal C)$, do not vanish for $|z|<1+\delta$.

Now that we have established that $\Psi_{n}(h_z,\mathcal C)^{1/n}$ is a well-defined analytic function of $z$, we proceed and prove that it is a normal family.  It follows from \eqref{eq:heine} and \eqref{eq:strongasympterm} that
\begin{multline}
|\Psi_n(h_z,\mathcal C)| =\left | \frac{D_n(e^{h_z} {\rm d} \mu)}{D_n({\rm d}\mu)} e^{-  \int h_z(w) K_n(w,w) {\rm d} \mu(w) }\right|
\leq  \frac{D_n(e^{\Re h_z} {\rm d} \mu)}{D_n({\rm d}\mu)} e^{-  \int \Re h_z(w) K_n(w,w) {\rm d} \mu(w) }\\=
 \Psi_n(\Re h_z,\mathcal C)=\Psi_n( h_{\Re z},\mathcal C).
\end{multline}
By combining this with \eqref{eq:bounddervativephiincommutatorsbounded} we see   that there exists an $M>0$ such that
$$|\Psi_n(h_z,\mathcal C)|^{1/n} \leq M,$$for $n \in \mathbb N$ and $|z| \leq 1+ \delta $  (observe that $\rho$ and $\eps$ can be chosen to work uniformly for the whole family of $h_{\Re z}$).   Hence, by Montel's Theorem,
$\Psi_{n}(h_z,\mathcal C)^{1/n},$
is a normal family of analytic functions on the disk $|z|< 1+ \delta$. This means that there exists a subsequence that converges uniformly to an analytic function on the disk. From \eqref{eq:hzreal0} we know that this limit must equal $1$ for $z \in \mathbb R$ and hence, by analyticity, it must equal $1$  for all $|z|<1+\delta$. In particular for  $z={\rm i}$ and this proves the statement.\end{proof}

\subsection{Proof of Theorem \ref{thm2}}\label{sec:proofthm2}
We now set, for $k_1 \leq k_2,$
$$Q_{k_1}^{k_2}= P_{k_2}-P_{k_1}.$$
The following  lemma is a variation on  \cite[Lem. 4.2]{BDjams} and is heavily based on the fact that  $\mathcal C$ is a banded matrix.
\begin{lemma} Let $h(z)= \sum_{|j| \leq H}  h_jz^j$ be a Laurent polynomial of degree $H \geq 1$. Set $M=2 (N+1) H$  and
\begin{equation}\label{eq:defCMM}
\mathcal C_M:= Q_{n-M}^{n+M} \calC Q_{n-M}^{n+M},
\end{equation}then
$\Phi_{n,N}(t,h,\mathcal C)=\Phi_{n,N}(t,h,\mathcal C_M)$. In the latter, we define $h(\mathcal C_M)$ as
$$h(\mathcal C_M): = \sum_{j \geq 0} h_j (\mathcal C_M)^j+ \sum_{j <0 } h_j (\mathcal C_M^* )^{-j}.$$
\end{lemma}
\begin{remark}
Observe that $\mathcal C_M$ also depends on $n$ and it would therefore be logical to write $\mathcal C_M^{(n)}$. However, we suppress the dependence on $n$ to avoid cumbersome notation.
\end{remark}
\begin{remark}
Note that $\mathcal C_M$ is not unitary. In fact, it is not invertible and hence negative powers of $\mathcal C_M$ do not make sense. The usual definition $h(\mathcal C_M)$ therefore fails. The alternative definition in the lemma is based on $\mathcal C^{-1}= \mathcal C^*$ and is well-suited for our purposes.
\end{remark}
\begin{proof}
From \eqref{eq:cumul} we see that $\Phi_{n,N}$ is a sum over terms of the form
\begin{equation} \label{eq:cumuldependsonsmallpart}
\Tr P_n h(\mathcal C)^{l_1}P_n h(\mathcal C)^{l_2} \cdots P_n h(\mathcal C)^{l_j} [h(\mathcal C),P_n]
\end{equation}
where $l_1+ \cdots + l_j = m \leq N$. Since $h$ is a Laurent polynomial and both $\calC$ and $\calC^{-1} =\calC^*$ are banded matrices we see that $h(\mathcal C)$ is banded. But then $[h(\mathcal C),P_n]$ is a very sparse matrix with the  only non-zero entries that are centered around the $nn$-entry. From the fact that we multiply $[h(\mathcal C),P_n]$ from the left with a number of banded matrices, all involving $\mathcal C$, it is not hard to see that $\Phi_{n,N}$ only depends on a relatively small part of $\calC$ that is concentrated around the $nn$-entry. The arguments below show this in a more precise and systematic way.

First of all, note that if $A$ is any banded matrix and $b$ is such that $A_{kl} =0$ if $|k-l|>b$, then
$$Q_{k_1}^{k_2} A= Q_{k_1}^{k_2} (Q_{m_1}^{m_2} A Q_{m_1}^{m_2}),$$
for any $m_1\leq k_1-b$ and $m_2\geq k_2+b$. Similarly, for any power $A^j$ with $j \geq 1$, $$Q_{k_1}^{k_2} A^j= Q_{k_1}^{k_2} (Q_{m_1}^{m_2} A Q_{m_1}^{m_2})^j,$$
for any $m_1\leq k_1-j b$ and $m_2\geq k_2+j b$. Since $\calC$ and $\calC^{-1}= \calC^*
$ are banded matrices with $b=2$,  this implies
$$Q_{k_1}^{k_2} h(\mathcal C)^{l_s}= Q_{k_1}^{k_2}h\left( Q_{m_1}^{m_2} \calC Q_{m_1}^{m_2}\right) ^{l_s},$$
for any $m_1\leq k_1-2 l_s  H$ and $m_2\geq k_2+2 l_s H$.   Then by taking $m_1= k_1-2 l_1  H$ and $m_2=k_2+2 l_1 H$, and using the fact that the projections commute, we find
\begin{align*}
Q_{k_1}^{k_2} P_n h(\mathcal C)^{l_1} P_n h(\mathcal C)^{l_2}&= Q_{k_1}^{k_2} P_n h(Q_{m_1}^{m_2} C Q_{m_1}^{m_2})^{l_1} P_n h(\mathcal C)^{l_2}\\
&=Q_{k_1}^{k_2} P_n h(Q_{m_1}^{m_2} \calC Q_{m_1}^{m_2})^{l_1} P_n Q_{m_1}^{m_2} h(\mathcal C)^{l_2}\\
&=Q_{k_1}^{k_2} P_n h(Q_{m_1}^{m_2} \calC Q_{m_1}^{m_2})^{l_1} P_n Q_{m_1}^{m_2} h( Q_{\tilde m_1}^{\tilde m_2}\calC Q_{\tilde m_1}^{\tilde m_2})^{l_2}\\
&=Q_{k_1}^{k_2} P_n h(Q_{m_1}^{m_2} \calC Q_{m_1}^{m_2})^{l_1} P_n  h( Q_{\tilde m_1}^{\tilde m_2}\calC Q_{\tilde m_1}^{\tilde m_2})^{l_2}\\
&=Q_{k_1}^{k_2} P_n h(Q_{\tilde m_1}^{\tilde m_2} \calC Q_{\tilde m_1}^{\tilde m_2})^{l_1} P_n  h( Q_{\tilde m_1}^{\tilde m_2}\calC Q_{\tilde m_1}^{\tilde m_2})^{l_2}
\end{align*}
for any $\tilde m_1\leq k_1-2(l_1+l_2) H$ and $\tilde m_2 \geq k_2+2(l_1+l_2)H$.
By iteration, we find,
\begin{equation}\label{eq:cumuldependsonsmallpart2}
Q_{k_1}^{k_2}  P_n h(\mathcal C)^{l_1} P_n h(\mathcal C)^{l_2} \cdots P_n h(\mathcal C)^{l_j}
= Q_{k_1}^{k_2}  P_n h( Q_{m_1}^{m_2} \calC Q_{m_1}^{m_2} )^{l_1} P_n h(Q_{m_1}^{m_2} \calC Q_{m_1}^{m_2} )^{l_2} \cdots P_n h(Q_{m_1}^{m_2} \calC Q_{m_1}^{m_2} )^{l_j},
\end{equation}
for any $l_j$ with $l_1+ \cdots l_j=m \leq N$, and $m_1 \leq k_1- 2 m H$ and $m_2  \geq k_2+ 2 m H$.

Now return to \eqref{eq:cumuldependsonsmallpart}. Since $h(\mathcal C)$ is banded  and $(h(\mathcal C))_{kl} =0$ if $|k-l| > 2 H$, we have that
\begin{equation}\label{eq:tja1}
[  h(\mathcal C) ,P_n] = [  h(\mathcal C) ,P_n]  Q_{n-2 H}^{n+2H},
\end{equation}
and
\begin{equation}\label{eq:tja2} [  h(\mathcal C) ,P_n] = [  h(\mathcal C_M) ,P_n] .
\end{equation}
By inserting \eqref{eq:tja1} into \eqref{eq:cumuldependsonsmallpart} and using fact that the trace is cyclic we find
 \begin{equation*}
 \Tr P_n h(\mathcal C)^{l_1}P_n h(\mathcal C)^{l_2} \cdots P_n h(\mathcal C)^{l_j} [h(\mathcal C),P_n]Q_{n-2 H}^{n+2H}
 = \Tr Q_{n-2 H}^{n+2H}P_n h(\mathcal C)^{l_1}P_n h(\mathcal C)^{l_2} \cdots P_n h(\mathcal C)^{l_j} [h(\mathcal C),P_n].
 \end{equation*}
 By substituting \eqref{eq:cumuldependsonsmallpart2} we can rewrite \eqref{eq:cumuldependsonsmallpart} further as
 \begin{multline}
\Tr Q_{n-2 H}^{n+2H} P_n h( Q_{m_1}^{m_2} \calC Q_{m_1}^{m_2} )^{l_1} P_n h(Q_{m_1}^{m_2} \calC Q_{m_1}^{m_2} )^{l_2} \cdots P_n h(Q_{m_1}^{m_2} \calC Q_{m_1}^{m_2} )^{l_j} [h(\mathcal C),P_n]
\\
 =\Tr  P_n h( Q_{m_1}^{m_2} \calC Q_{m_1}^{m_2} )^{l_1} P_n h(Q_{m_1}^{m_2} \calC Q_{m_1}^{m_2} )^{l_2} \cdots P_n h(Q_{m_1}^{m_2} \calC Q_{m_1}^{m_2} )^{l_j} [h(\mathcal C),P_n] Q_{n-2 H}^{n+2H}\\
 = \Tr P_n h( Q_{m_1}^{m_2} \calC Q_{m_1}^{m_2} )^{l_1} P_n h(Q_{m_1}^{m_2} \calC Q_{m_1}^{m_2} )^{l_2} \cdots P_n h(Q_{m_1}^{m_2} \calC Q_{m_1}^{m_2} )^{l_j} [h(\mathcal C),P_n]
 \end{multline}
for any $m_1 \leq n-2 (m+1) H$ and $m_2 \geq n+2 (m+1)H$. Now by taking $m_1=n- 2(N+1)H$ and $m_1=n+ 2(N+1)H$, we obtain that  \eqref{eq:cumuldependsonsmallpart} can be written as
$$ \Tr P_n h(\mathcal C_M)^{l_1} P_n h(\mathcal C_M )^{l_2} \cdots P_n h(\mathcal C_M )^{l_j} [h(\mathcal C),P_n],$$
where $\mathcal C_M$ is as in the statement of the theorem. Finally, by \eqref{eq:tja2} we  replace the commutator in the product and this finishes the proof.
\end{proof}

We are now ready for the
\begin{proof}[Proof of Theorem \ref{thm2}]

Let $\mathcal C$ and $\tilde{\mathcal  C}$ be the two CMV matrices corresponding to $\mu$ and $\tilde \mu$.

We start by supposing that $h$ is real and first prove that
\begin{equation}\label{eq:thm2forphin} \left|\Phi_{n_j}(t,h,\mathcal C)-\Phi_{n_j}(t,h,\tilde {\mathcal C})\right| \to 0,
\end{equation}
as $j \to \infty$. The proof of the statement then follows from $\Psi_n(h,\mathcal C)= \exp(\Phi_n(1,h,\mathcal C))$.

For $H \in \mathbb N$, we define $h^H (z)= \sum_{|j|\leq H} h_j z^j$ and write
\begin{multline}\label{eq:longsplit}
\left|\Phi_{n}(t,h,\mathcal C)-\Phi_{n}(t,h,\tilde {\calC})\right|
 \leq  \left|\Phi_{n}(t,h,\mathcal C)-\Phi_{n}(t,h^H,\mathcal C)\right| \\
+ \left|\Phi_{n}(t,h^H,\mathcal C)-\Phi_{n,N}(t,h^H,\mathcal C)\right| +\left|\Phi_{n,N}(t,h^H,\mathcal C)-\Phi_{n,N}(t,h^H,\tilde {\mathcal  C})\right| \\+\left|\Phi_{n}(t,h^H,\tilde{\calC})-\Phi_{n,N}(t,h^H,\tilde {\mathcal C})\right| + \left|\Phi_{n}(t,h^H,\tilde{\calC})-\Phi_{n}(t,h,\tilde{\mathcal  C})\right|.
\end{multline}
We start by estimating $\left|\Phi_{n,N}(t,h^H,\mathcal C)-\Phi_{n,N}(t,h^H,\tilde {\mathcal  C})\right|$. We recall that for $M$ large enough we have $\Phi_{n,N}(t,h^H,\mathcal C)=\Phi_{n,N}(t,h^H,\mathcal C_M)$ and the latter is a sum over terms
$$\Tr P_n h^H(\mathcal C_M)^{l_1}P_n h^H(\mathcal C_M)^{l_2} \cdots P_n h^H(\mathcal C_M)^{l_j} [h^H(\mathcal C_M),P_n]$$
(cf. \eqref{eq:cumuldependsonsmallpart}) and similarly for $\Phi_{n,N}(t,h^H,\tilde {\mathcal C})=\Phi_{n,N}(t,h^H,\tilde{\mathcal C}_M)$. To estimate the difference between the values of these terms for $\mathcal C_M$ and $\tilde {\mathcal C}_M$, we note that the trace of a finite rank matrix is dominated by the rank times the operator norm. Since the rank of $\mathcal C_M$ and $\tilde{\mathcal C}_M$ is $2M$ we thus have
\begin{multline}
\left|\Tr P_n h^H(\mathcal C_M)^{l_1}P_n h^H(\mathcal C_M)^{l_2} \cdots P_n h^H(\mathcal C_M)^{l_j} [h^H(\mathcal C_M),P_n]\right.
\\
- \left. \Tr P_n h^H(\tilde{\mathcal C}_M)^{l_1}P_n h^H(\tilde{\mathcal C}_M)^{l_2} \cdots P_n h^H(\tilde{\mathcal C}_M)^{l_j} [h^H(\tilde{\mathcal C}_M),P_n]\right|,\\
\leq  2  M\left\| P_n h^H(\mathcal C_M)^{l_1}P_n h^H(\mathcal C_M)^{l_2} \cdots P_n h^H(\mathcal C_M)^{l_j} [h^H(\mathcal C_M),P_n]\right.
\\
- \left.  P_n h^H(\tilde{\mathcal C}_M)^{l_1}P_n h^H(\tilde{\mathcal C}_M)^{l_2} \cdots P_n h^H(\tilde{\mathcal C}_M)^{l_j} [h^H(\tilde{\mathcal C}_M),P_n]\right\|_\infty.
\end{multline}
We now  replace each $\mathcal C_M$  in the first term by $\tilde{\mathcal C}_M$ step by step and estimate all the terms we  obtain this way. To this end, we note that since $\|\mathcal C\|_\infty= \|\tilde {\mathcal C}\|_\infty =1$ we also have $\|\mathcal C_M \|_\infty, \|\tilde{\mathcal C}_M\|_\infty \leq 1$ and thus
\begin{multline*}
\|h^H (\mathcal C_M) -h^H(\tilde{\mathcal C}_M)\|_\infty  \leq  \sum_{j=0}^H |h_j |  \|\mathcal C_M^j -\tilde{\mathcal C}_M^j\|_\infty+ \sum_{j=1}^H |h_{-j} |  \|(\mathcal C_M^*)^j -(\tilde{\mathcal C}_M^*)^j\|_\infty \\ \leq \sum_{|j| \leq H} |j| |h_j |  \|\mathcal C_M-\tilde{\mathcal C}_M\|_\infty \leq \sqrt H \|h^H\|_{ \mathfrak B _{\frac12}}   \|\mathcal C_M-\tilde{\mathcal C}_M\|_\infty.
\end{multline*}
In the end, the result is
\begin{multline}
\left|\Tr P_n h^H(\mathcal C_M)^{l_1}P_n h^H(\mathcal C_M)^{l_2} \cdots P_n h^H(\mathcal C_M)^{l_j} [h^H(\mathcal C_M),P_n]\right.
\\
- \left. \Tr P_n h^H(\tilde{\mathcal C}_M)^{l_1}P_n h^H(\tilde{\mathcal C}_M)^{l_2} \cdots P_n h^H(\tilde{\mathcal C}_M)^{l_j} [h^H(\tilde{\mathcal C}_M),P_n]\right| \leq c \|\mathcal C_M-\tilde{\mathcal C}_M\|_\infty,
\end{multline}
where $c$ is a constant that depends on $H$ and $N$ and $\|h\|_{ \mathfrak B _{\frac12}}$. Observe that $\mathcal C_M$ and $\tilde{\mathcal C}_M$ depend on $n$ and that $ \|\mathcal C_M-\tilde{\mathcal C}_M\|_\infty\to 0$ along the subsequence $\{n_j\}_j$ by the assumption in the theorem. Hence we have, for $M$ large enough,
\begin{equation} \label{eq:phinNCtildC}
\lim_{j\to \infty} \left|\Phi_{n_j,N}(t,h^H,\mathcal C)-\Phi_{n_j,N}(t,h^H,\tilde \calC )\right|=
\lim_{j\to \infty} \left|\Phi_{n_j,N}(t,h^H,\mathcal C_M)-\Phi_{n_j,N}(t,h^H,\tilde {\calC}_M)\right|= 0,
\end{equation}
for any fixed $N$ and $H$.

Starting from \eqref{eq:longsplit} and using \eqref{eq:phinNCtildC} and Lemmas \ref{lem:bound1} and \ref{lem:bound2} we find
\begin{equation*}
\limsup_{j\to \infty} \left|\Phi_{n_j,N}(t,h^H,\mathcal C)-\Phi_{n_j,N}(t,h^H,\tilde {\calC})\right|
 \leq  c_1 \|h-h^H\|_{ \mathfrak B _{\frac12}} +c_2 \sum_{m=N+1}^\infty m^{3/2} 2^{-m},
\end{equation*}
for $|t|  \leq 1/(2 {\rm e} \|h\|_{\mathfrak B_\frac{1}{2}})$ (note that $\|h\|_{\mathfrak B_\frac{1}{2}} \geq   \|h^H\|_{\mathfrak B_\frac{1}{2}} \geq \|h^H\|_{\infty} $), where $c_1$ and $c_2$ are constants that depend on $h$ but not on $N$ and $H$. By taking $N, H \to \infty$ we indeed obtain the statement and this finishes the proof of \eqref{eq:thm2forphin} for $t$ in a neighborhood of the origin.

Next we  prove that \eqref{eq:thm2forphin} also holds for $t =1$. Note that  $\{\Phi_{n_j}(t,h,\mathcal C)-\Phi_{n_j}(t,h,\tilde {\mathcal C})\}_{j \in \mathbb N}$ is a normal family for $t \in \cup_{x \in [0,1]} B_{x,\eps}$. Hence there exists a  subsequence that $\{n_{j_\ell}\}$  along which the family converges to an analytic function. Since we have \eqref{eq:thm2forphin} in a neighborhood of the origin, we know that this function must be  identically zero in that neighborhood. By analyticity it is zero for $t \in \cup_{x \in [0,1] }B_{x,\eps}$. This proves that we indeed have \eqref{eq:thm2forphin} for any  $t \in \cup_{x \in [0,1] }B_{x,\eps}$ and in particular for $t=1$.

To finish the proof of  Theorem \ref{thm2} for real-valued $h$, observe that   by \eqref{eq:relationPsiandPhi} and \eqref{eq:bounddervativephiincommutators}  there exists an $M>0$ such that
$$\left|\Psi_n(h,\mathcal C)\right| \leq M,$$
for $n\in \mathbb N$ and any  CMV matrix $\calC$. Hence if $\calC$ and  $\tilde \calC$ are two CMV matrices, then
\begin{equation*}
\left|\Psi_n(h,\mathcal C)-\Psi_n(h,\tilde \calC)\right| = \left|\Psi_n(h,\mathcal C)\right|\left|1-\frac{\Psi_n(h,\tilde \calC)}{\Psi_n(h,\mathcal C)}\right|
 \leq M \left|1-{\rm e}^{\Phi_n(1,h,\tilde \calC)-\Phi_n(1,h,\mathcal C)}\right|.
\end{equation*}
The statement for real-valued $h$ therefore follows from \eqref{eq:thm2forphin}.

The extension to  complex-valued $h$ is  analogous  to the argument given in the proof of Proposition \ref{prop:weak}. We define
$$h_z= \Re h + z \Im h,$$
so that $h=h_{\rm i}$. From the proof for the real-valued case we learn that for $z \in \mathbb R$ we have
\begin{equation}\label{eq:hzreal}
\Psi_{n_j}(h_z,\mathcal C)- \Psi_{n_j}(h_z,\tilde{\mathcal C})\to 0.
\end{equation}
 It remains to prove that this also holds for $z= {\rm i}$.

 It follows from \eqref{eq:heine} and \eqref{eq:strongasympterm} that

\begin{multline}
|\Psi_n(h_z,\mathcal C)| =\left | \frac{D_n(e^{h_z} {\rm d} \mu)}{D_n({\rm d}\mu)} e^{-  \int h_z(w) K_n(w,w) {\rm d} \mu(w) }\right|
\leq  \frac{D_n(e^{\Re h_z} {\rm d} \mu)}{D_n({\rm d}\mu)} e^{-  \int \Re h_z(w) K_n(w,w) {\rm d} \mu(w) }\\=
 \Psi_n(\Re h_z,\mathcal C)=\Psi_n( h_{\Re z},\mathcal C).
\end{multline}
By combining this with \eqref{eq:bounddervativephiincommutators} we see   that there exists an $M>0$ such that
$$|\Psi_n(h_z,\mathcal C)| \leq M,$$for $n \in \mathbb N$, $|z| \leq 2 $ and any  CMV matrix $\mathcal C$.   Hence, by Montel's Theorem,
$$\Psi_{n_j}(h_z,\mathcal C)- \Psi_{n_j}(h_z,\tilde{\mathcal C}),$$
is a normal family of analytic function on the disk $|z|\leq 2$. This means that there exists subsequence such that converges uniformly to an analytic function on the disk. From \eqref{eq:hzreal} we know that this limit must vanish for $z \in \mathbb R$, hence it must be the zero function for all $|z|\leq 2$. This means that we have \eqref{eq:hzreal} for all $|z| \leq 2$. In particular, for  $z={\rm i}$ and we proved the statement.
\end{proof}

\subsection{Proof of Theorem \ref{thm:rightlimit}}\label{sec:proofrightlimit}
Part of the conclusion of Theorem \ref{thm:rightlimit} is that the limit $q(h)$ is positive for real-valued $h$.  Therefore we can write
$$q(h)= {\rm e}^{Q(h)},$$
for some function $Q$ and real valued $h$.
Before we come to the proof of Theorem \ref{thm:rightlimit}, we first present an expression for  $Q(h)$.  To this end, we need the right limit of the CMV matrix, which is the double infinite matrix  given by the following limit
$$( \mathcal C^R)_{k\ell}=\lim_{j \to \infty} (\mathcal C)_{n_j+k,n_j+\ell}, \qquad k,\ell \in \bbZ.$$
Here $\{n_j\}_j$ is the sequence for which $\alpha_{n_j+k}  \to  \beta_k$. Then we define $F_m(A)$ for a $\bbZ \times \bbZ$ matrix $A$ by
\begin{equation}
\label{eq:cumullim}
F_m(A)=
\frac{1}{m+1}\sum_{j=1}^m {(-1)^{j-1}}\sum_{l_1+ \cdots + l_j=m, l_i \geq 1} \frac{\Tr P_- A^{l_1}P_-A^{l_2} \cdots P_- A^{l_j} [A,P_-]}{l_1! \ldots l_j!},
\end{equation}
where $P_-$ is the projection operator on $\ell_2(\bbZ)$ that projects on the negative part of $\bbZ$, i.e.
$$(P_- x)_k= \begin{cases}
x_k, & k<0\\
0, & \text{ otherwise.}
\end{cases}$$
Note that  $F_m(A)$ is well-defined for banded matrices  $A$, since in that case $[A,P_-]$ has only finitely many non-zero entries. If $\|A\|_\infty \leq 1$ then the same arguments that showed that $E_m^{(n)}(h(\mathcal C))$ is well-defined, also show that $F_m(h(A))$  is well-defined, with $h(A)= \sum_{j \geq 0} h_j A^j+ \sum_{j >0} h_{-j} (A^*) ^j$   where $h \in  \mathfrak B _{\frac{1}{2}}$. Moreover, all the boundedness and continuity properties of $E_m^{(n)}(h(\mathcal C))$ that we proved also hold for $F_m(h(A))$. We summarize this in the following lemma.
\begin{lemma} For $h \in  \mathfrak B _{\frac12}$ and a banded $\mathbb Z \times \mathbb Z$ matrix $A$ with $\|A\|_\infty \leq 1$ we define
$$\Xi (t,h,A)= \sum_{m=1}^\infty t^{m+1} F_m(h(A)).$$
Then $t \to \Xi(t,h,A)$ is a well-defined analytic function in a sufficiently small neighborhood of the origin. Moreover, there exists a constant $c>0$ such that
$$|\Xi(t,h_1,A)-\Xi(t,h_2,A)| \leq c \|h_1-h_2\|_{{ \mathfrak B }_{\frac{1}{2}}},$$
for $t$ in a sufficiently  small neighborhood of the origin and $A$ such that $\|A\|_\infty\leq 1$.
\end{lemma}

The relation between  $E_m^{(n)}$ and $F_m$ is as follows. First we embed the space matrices of $\mathbb N \times \mathbb N$ into the space of $\bbZ \times \bbZ$ matrices by adding zero-entries. Moreover, we extend $P_n$ to an operator on $\ell_2(\bbZ)$ by $(P_nx)_k=x_k$ if $k \leq n$ and by $(P_nx)_k=0$ if $k > n$. Then we can view the operators in the traces in the definition \eqref{eq:cumul} of $E_m^{(n)}$ as operators on $\ell_2(\bbZ)$.
Next, we use  the shift operator $\mathcal S_n$ that maps sequences $x=\{x_k\}_{k \in \bbZ}$ to $\mathcal S_n x= \{x_{n+k}\}_{k \in \bbZ}$. Then $(\mathcal S_n^*x)_k= x_{n-k}$ and  $\mathcal S_n^* P_n \mathcal  S_n= P_-$.  Then it follows by the fact that the trace is cyclic and $\mathcal S_n^*h(\mathcal C)S_n=h(\mathcal S_n^* \mathcal  C  \mathcal S_n)$ that
\begin{equation}\label{eq:fromDtoE}
E_m^{(n)} (h(\mathcal C))= F_m(h(\mathcal S_n^* \mathcal  C  \mathcal S_n)).
\end{equation}

Next we introduce the truncation of the right limit $\mathcal C^R$ defined by
$$\left(\mathcal C_M^R\right)_{jk}=\begin{cases}
(\mathcal C^R)_{j k}, & -M+1\leq j,k\leq M,\\
0, &\text{otherwise},
\end{cases}
$$
and the map  $t \mapsto Q_M(t h)$ given by
$$Q_M(t h)= \sum_{m=1}^\infty t^{m+1}  F_m(h(\mathcal C^R_M)),$$
which defines an analytic function for $t$ in a sufficiently small neighborhood of the origin.  From the proof below we find that the limit\begin{equation}\label{eq:limitQth}
Q(t h)=\lim_{M\to \infty} \sum_{m=1}^\infty t^{m+1}  F_m(h(\mathcal C^R_M)),
\end{equation}
is  a well-defined analytic function for $t$ in a sufficiently small neighborhood of the origin.
Moreover,  $t \mapsto Q(th)$ can be extended to an analytic function on $\cup_{x \in [0,1]} B_{x,\eps}$ for $\eps>0$. Then at $t=1$ we find the value $q(h)={\rm e}^{Q(h)}$ in Theorem~ \ref{thm:rightlimit} for real-valued functions $h$.

\begin{proof}[Proof of Theorem \ref{thm:rightlimit}]
We will first assume that $h$ is real-valued. We will also assume without of loss of generality that the right limit is along the trivial sequence $n_j=j$ so that
$$(\mathcal C^R)_{k\ell} =  \lim_{n\to \infty} ( \mathcal C )_{n+k ,n+\ell}.$$

Let us first consider the case where $h(z)= \sum_{|j| \leq H} h_j z^j$ is a Laurent polynomial.

We expand  again $
\Phi_{n}(t,h,\mathcal C)= \sum_{m=1}^\infty t^{m+1}E_m^{(n)}(h(\mathcal C))
$
and note that
$E_m^{(n)}(h(\mathcal C))=E_m^{(n)}(h(\mathcal C_M))$ for $m \leq
M/(2H)-1$, where $\mathcal C_M$ is as defined in \eqref{eq:defCMM}.
Then by  \eqref{eq:fromDtoE} and an argument similar as in the proof of Theorem \ref{thm2} and using
$$\lim_{n\to \infty} h(\mathcal S_n^* \mathcal  C_M  \mathcal S_n))= h(\mathcal C_M^R),$$
we easily find
$$\lim_{n\to \infty} E_m^{(n)}(h(\mathcal C))=\lim_{n\to \infty} E_m^{(n)}(h(\mathcal C_M))=
\lim_{n\to \infty}  F_m(h(\mathcal S_n^* \mathcal  C_M \mathcal  S_n))= F_m(h(\mathcal C_M^R)),$$
for $m \leq
M/(2H)-1$.  Now it is also important to note that the left-most term does not depend on $M$ and hence none of the terms do and they  hold whenever $m \leq
M/(2H)-1$. Hence we have
$$\lim_{n\to \infty} E_m^{(n)}(h(\mathcal C))= \lim_{M \to \infty} F_m(h(\mathcal C_M^R)),$$ for all $m \in \mathbb N$.

By Lemma \ref{lem:normalfamily} we know that $t\mapsto \Phi_n(t,h, \mathcal C)$ defines a normal family of analytic functions on  $\cup_{x \in [0,1]} B_{x,\eps}$ with $\eps$ as in Lemma \ref{lem:normalfamily}. Hence there exists a convergent subsequence $\Phi_{n_k}$ with analytic limit $\Phi$. We willl show that  $\Phi$ does not depend on the subsequence after which the statement follows. Indeed,  from the above we know that $\Phi$ has series expansion around the origin with coefficients
$$\lim_{M \to \infty} F_m(h(\mathcal C_M^R)),$$
which does not depend on the precise subsequence $n_k$.  This proves that $\Phi(t,h,\mathcal C)= Q(th)$
as given in \eqref{eq:limitQth} and hence we obtain \eqref{eq:generalSSLT1} for Laurent polynomials $h$.

The extension from Laurent polynomials to general $h \in  \mathfrak B _{\frac{1}{2}}$ follows by a straightforward argument  similar to the one in the proof of Theorem \ref{thm2} and is left to the reader.  The statement that $h \mapsto Q(h)$ is continuous with respect to $\|\cdot\|_{ \mathfrak B _{\frac{1}{2}}}$ follows from  the  fact that by \eqref{eq:equicontinuitysobolev} the family of functions $\{\Phi_{n}\}_n$ is equicontinuous with respect to this norm.

The case of complex valued $h$ can be shown by using an argument based on Montel's Theorem, very similar to  the proof  Theorem \ref{thm2}. However, in this argument we may loose positivity of the limit and therefore we can no longer write $q(h)={\rm e}^{Q(h)}$ for complex-valued functions. Since the argument is almost identical to the argument in the proof of Theorem \ref{thm2} we leave the details  to the reader and this concludes the proof. \end{proof}

\section{Proof of Theorem \ref{thm1} and Proposition \ref{prop:alpha}}\label{sec:proofthm1}
 By Theorem \ref{thm2} we see that two (families of) CMV matrices which have the same right limit, also have the same limit for the ratio \eqref{eq:strongasympterm} (if exists). To prove Theorem \ref{thm2} it is therefore sufficient to analyze a particular CMV matrix with $\alpha_j$. In this section we will therefore analyze the simplest case, namely
 $$\alpha_j=\alpha,$$
 for some $\alpha$ in the unit disk.
We will prove that for such $\calC$ and for $ h \in \mathfrak B_{\frac12}$, we have
$$\lim_{n\to\infty} \Psi_n(h, \calC) = \lim_{n\to\infty} \det (I+P_n ({\rm e}^{h(\mathcal C)}-I)P_n) {\rm e}^{-\Tr P_n h(\mathcal C) P_n} = {\rm e}^{Q_\alpha(h)},$$
where $Q_\alpha(h)$ is given in~\eqref{eq:defQa}.

\subsection{Preliminaries}

The idea is to use an identity due to Ehrhardt \cite[Th. 2.2]{E}. He proved that if $A,B$ are two operators for which  the commutator $[A,B]$ is trace class, then
$$\det {\rm e}^{-A } {\rm e}^{A+B} {\rm e}^{-B} = {\rm e}^{-\frac12\Tr [A,B]}.$$
The left-hand side shoud be understood as a Fredholm determinant for the operator ${\rm e}^{-A } {\rm e}^{A+B} {\rm e}^{-B} -I$. So part of the statement is that the latter operator is trace class if $[A,B]$ is trace class.

The following principle is the key to the proof of Theorem \ref{thm1}.

\begin{proposition}\label{prop:commutator}
Let ${\mathcal C}$ be the CMV matrix and $h \in \mathcal B_{\frac12}$, such that there exist $U$ and $L$ satisfying
\begin{enumerate}
\item[(i)] $h(\mathcal C)= L+U$,
\item[(ii)] $L$ is lower triangular and $U$ is upper triangular,
\item[(iii)] $[L,U]$ is of trace class.
\end{enumerate}
Then
$$\det (I+ P_n ({\rm e}^{h(\mathcal C)}-I)P_n) {\rm e}^{-\Tr P_n h(\mathcal C) P_n} \to  {\rm e}^{\frac{1}{2}  \Tr [U,L]},$$
as $n \to \infty$.
\end{proposition}

\begin{proof}
By the triangularity of $U$ and $L$ we have
\begin{equation}\label{eq:triangularity}
P_n L P_n=  P_n L \quad \text{ and }  \quad  P_n U P_n=  U P_n.
\end{equation}
Hence we also have
$$ {\rm e}^{-P_n L P_n}= Q_n+P_n {\rm e}^{-L }P_n, \quad \text{ and }  \quad  {\rm e}^{-P_n U P_n}= Q_n+ P_n {\rm e}^{-U}P_n,$$
and therefore
$$ {\rm e}^{-\Tr P_n L P_n}=\det \left( Q_n+ P_n {\rm e}^{-L}P_n\right), $$
and $$  {\rm e}^{-\Tr P_n U P_n}=\det \left( Q_n+ P_n{\rm e}^{-U}P_n\right).$$
After some simple algebra we find
\begin{multline*}
\det \left( Q_n+ P_n{\rm e}^{h(\mathcal C)}P_n\right) {\rm e}^{-\Tr P_n h(\mathcal C) P_n}
={\rm e}^{-\Tr P_n L P_n}\det \left( Q_n+ P_n{\rm e}^{h(\mathcal C)}P_n\right) {\rm e}^{-\Tr P_n U P_n} \\
=\det \left( Q_n+ P_n {\rm e}^{-L}P_n\right)\det \left( Q_n+ P_n{\rm e}^{h(\mathcal C)}P_n\right) \det \left( Q_n+ P_n{\rm e}^{-U}P_n\right)\\
=\det \left( Q_n+ P_n   {\rm e}^{-L}P_n{\rm e}^{h(\mathcal C)} P_n{\rm e}^{-U}P_n\right).
\end{multline*}
Now use \eqref{eq:triangularity}  again to deduce that
\begin{multline*}
\det \left( Q_n+ P_n{\rm e}^{h(\mathcal C)}P_n\right) {\rm e}^{-\Tr P_n h(\mathcal C) P_n}=\det \left( Q_n+ P_n   {\rm e}^{-L}{\rm e}^{h(\mathcal C)} {\rm e}^{-U}P_n\right)\\
=\det \left( I+ P_n \left(  {\rm e}^{-L}{\rm e}^{h(\mathcal C)} {\rm e}^{-U}-I\right)P_n\right).
\end{multline*}
We recall that if $A$ is trace class then $P_nA P_n \to A$ in trace norm. Moreover, the Fredholm determinant is continuous with respect to the trace norm \cite{SimonTrace}. Hence by taking the limit $n \to \infty$ we find the statement. \end{proof}

The latter proposition works for any decomposition $h(\mathcal C)=L+U$, but it is not difficult to see we do not have much freedom. The off-diagonal entries of $U$ and $L$ are fixed by $h(\mathcal C)$ and we only have freedom for the diagonal entries. This freedom we will need to make sure that the commutator $[U,L]$ is trace class. As we will see, we only have a trace class commutator for a very particular choice in the diagonal entries.   The construction of $U$ and $L$ will take the rest of this section.

\subsection{Unwrapping of the CMV matrix}

 It will be more illuminative for our purposes to ``unwrap'' the structure of a CMV matrix as follows. Let $R$ be the isometry  $\ell^2(\bbN)\to\ell^2(\bbZ)$ (viewed as a $\bbZ\times \bbN$ matrix) defined via
\begin{equation}\label{eq:R}
(R)_{jk} =
\begin{cases}
1 & \mbox{if } j>0, k=2j-1, \\
1 & \mbox{if } j\le 0, k=-2(j-1), \\
0 & \mbox{otherwise}
\end{cases}
\end{equation}
for $j\in\bbZ,k\in\bbN$.
Direct calculation shows that $$\calD:=R\calC R^* : \ell^2(\bbZ)\to\ell^2(\bbZ)$$ takes the block form
\begin{equation}\label{calD}
\calD=
\left(
\begin{array}{cc}
J \calD^*_{11} J^* & J \calD_{12} \\
- \calD^*_{21} J^* & \calD_{22}
\end{array}
\right),
\end{equation}
where $J: \ell^2(\bbN) \to \ell^2(\bbZ\setminus\bbN)$ is the isometry taking the $k$-th standard unit vector $e_k$ of $\ell^2(\bbN)$ into $e_{-k+1}$ of $\ell^2(\bbZ\setminus\bbN)$, and $\calD_{jk}:\ell^2(\bbN)\to\ell^2(\bbN)$ are two-diagonal matrices
\begin{align}
\label{pt1}
\calD_{11} &= \left(
\begin{array}{ccccc}
-\bar{\alpha}_0 {\alpha}_1 &  \rho_1 \rho_2 & 0  & 0  &   \\
0 &  -\bar{\alpha}_2 {\alpha}_3  &  \rho_3 \rho_4 & 0 &   \\
0 &  0 &  -\bar{\alpha}_4  {\alpha}_5 & \rho_5\rho_6 &    \\
0 & 0 & 0 & -\bar{\alpha}_6  {\alpha}_7 & \ldots \\
 \hphantom{-\alpha_3  \bar{\alpha}_4 } &   &   &  & \ldots  \\
\end{array}
\right),
\\
\label{pt2}
\calD_{22} & = \left(
\begin{array}{ccccc}
 \bar{\alpha}_0 & \rho_0 \rho_1 & 0  &  0  &   \\
0 &  -\alpha_1 \bar{\alpha}_2  &  \rho_2 \rho_3 & 0 &   \\
0 &  0 &  -\alpha_3  \bar{\alpha}_4 & \rho_4\rho_5 &    \\
0 & 0 & 0 & -\alpha_5  \bar{\alpha}_6 & \ldots \\
 \hphantom{-\alpha_3  \bar{\alpha}_4 } &   &   &  & \ldots  \\
\end{array}
\right),
\\
\label{pt3}
\calD_{21} &= \left(
\begin{array}{ccccc}
 -\alpha_1 \rho_0 & -\alpha_2 \rho_1 &  0 &  0  &   \\
0 &  - \alpha_3 \rho_2   &  -\alpha_4 \rho_3 & 0 &   \\
0 &  0  &  - \alpha_5 \rho_4  & -\alpha_6 \rho_5 &    \\
0 & 0 & 0  &  -\alpha_7 \rho_6  & \ldots \\
 \hphantom{-\alpha_3  \bar{\alpha}_4 } &   \hphantom{-\alpha_3  \bar{\alpha}_4 } &   \hphantom{-\alpha_3  \bar{\alpha}_4 } &  \hphantom{-\alpha_3  \bar{\alpha}_4 } & \ldots  \\
\end{array}
\right),
\\
\label{pt4}
\calD_{12} &= \left(
\begin{array}{ccccc}
\rho_0 & -\alpha_0 \rho_1 &  0 &  0  &   \\
0 &  -\alpha_1 \rho_2  &  -\alpha_2 \rho_3 & 0 &   \\
0 &  0 &  -\alpha_3 \rho_4 & -\alpha_4 \rho_5 &    \\
0 & 0& 0 & -\alpha_5 \rho_6 & \ldots \\
 \hphantom{-\alpha_3  \bar{\alpha}_4 } &   \hphantom{-\alpha_3  \bar{\alpha}_4 } &   \hphantom{-\alpha_3  \bar{\alpha}_4 } &  \hphantom{-\alpha_3  \bar{\alpha}_4 } & \ldots  \\
\end{array}
\right).
\end{align}

Note that when $\alpha_j\equiv \alpha$ for all $j$ then (ignoring the $(1,1)$-entry) each of the operators~\eqref{pt1}--\eqref{pt4} is Toeplitz. This motivates us to introduce the following notation.

%
%
%

If $s(z)$ and $t(z)$ are Laurent polynomials, we define
$$
\DT(s,t):=
\left(
\begin{array}{cc}
J \T(s)^* J^* & J \T(t) \\
-\T(t)^* J^* & \T(s)
\end{array}
\right).
$$
Similarly, if $s(z)$, $t(z)$, $p(z)$, $q(z)$ are  Laurent polynomials, we define
$$
\QT(s,t,p,q):=
\left(
\begin{array}{cc}
J \T(s)^* J^* & J \T(t) \\
-\T(p)^* J^* & \T(q)
\end{array}
\right).
$$

In particular, $\DT(s,t) = \QT(s,t,t,s)$, of course.

For a function $s(z)$ of a complex variable $z$, we will occasionally use a shortcut $zs$ to denote the function $z s(z)$. We also denote
\begin{align}
s^*(z)&:= \overline{s(1/\bar z)},\\
\tilde{s}(z)&:=s(1/z),\\
\bar{s}(z)&:= \overline{s(\bar{z})}.
\end{align}

In the next lemma we show that each family of $\T$, $\DT$, $\QT$  matrices forms an algebra with respect to the usual matrix multiplication and addition if we agree to ignore finite rank perturbations. The key to these results is the following well-known identity relating Toeplitz and Hankel determinants,
\begin{equation}\label{hankel}
\T(s)\T(t) = \T(s t) - \Ha(s) \Ha(\tilde t).
\end{equation}

\begin{lemma}\label{lemDT}
For two operators $A$ and $B$ let us write $A\stackrel{f.e.}{=} B$ if $A-B$ has finitely many non-zero entries in the standard basis.
\begin{itemize}
\item[(i)] 
The following equalities hold:
\begin{align*}
&\DT(s_1,t_1)   \DT(s_2,t_2)  \stackrel{f.e.}{=}  \DT(s_1 s_2 -\tb t_1 t_2, \tb s_1 t_2 + t_1 s_2), \\
&\QT(s_1,t_1,p_1,q_1)   \QT(s_2,t_2,p_2,q_2)   
\stackrel{f.e.}{=}  \QT(s_1 s_2 -\tb t_1 p_2, \tb s_1 t_2 + t_1 q_2, p_1 s_2+\tb q_1 p_2, -\tb p_1 t_2 + q_1 q_2).
\end{align*}

\item[(ii)] The following identities hold: $T(s)^* = T(s^*)$, $\DT(s,t)^* = \DT(\tb s, -t)$ and $\QT(s,t,p,q)^* = \QT(\tb s,-p,-t,\tb q)$.

\end{itemize}
\end{lemma}
\begin{proof}
(i) is immediate from the definition and the fact that $\T(s)\T(t)-\T(st)$ is of finite rank (see~\eqref{hankel}). (ii) is immediate. 
\end{proof}

 \subsection{Unwrapping $h(\calC)$}

Let $h$ be a Laurent polynomial. The main aim of this section is to 
understand the structure of $h(\calD)$, see Corollary~\ref{cor:hcqt} below.

   As we saw earlier, $\calD$ and its integer powers too (by Lemma~\ref{lemDT}) have the $\DT$ structure (up to finitely many entries). Thus we may write
   \begin{equation}\label{eq:powersD}
   \calD^k \fr \DT(s_k,t_k)
   \end{equation}
   for some Laurent polynomials $s_k$ and $t_k$. 
Instead of working with the symbols $t_k$'s, it will actually be convenient to remove the phase by introducing
\begin{equation}\label{vk}
v_k(z) := \frac{\bar \alpha}{|\alpha|} t_k(z).
\end{equation}

Trivially, $s_0 = 1$, $v_0 = 0$. Since $\calD$ is unitary, we get $\calD^{-k} = (\calD^k)^*$ for all $k\in\bbN$, so by Lemma~\ref{lemDT}(ii), $s_{-k} = \tb s_k$ and $v_{-k} = -v_k$.

From~\eqref{calD}, we obtain 
\begin{align}
\label{s1} s_1(z) & = -|\alpha|^2  + \rho^2 \frac{1}{z}, \\
\label{v1} v_1(z) & = - {|\alpha|} \rho \left(1+\frac{1}{z} \right),
\end{align}
where $\rho:=\sqrt{1-|\alpha|^2}$.

In the following lemma we collect properties of $s_k$'s and $v_k$'s.

\begin{lemma}\label{lemSkvk}
\begin{itemize}
\item[(i)] For all $k\in\bbZ$, $s_k$ and $v_k$ have real coefficients, that is,
 \begin{align}
\label{skvk2} \overline{ {s}_k(\bar z) } &= s_k(z), \\
\label{skvk3} \overline{ {v}_k(\bar z) } & = v_k(z).
 \end{align}
\item[(ii)] For all $k\in\bbZ$:
 \begin{align}
\label{skvk1} v_k(z) &= -\tfrac{|\alpha|}{\rho} \frac{\tb{s}_k(z)-s_k(z)}{z-1}, \\
\label{skvk4}  \tb v_k(z) &= z v_k(z).
 \end{align}
\item[(iii)] For all $k\in\bbN$:
\begin{equation}
\label{recurrence}
\left( \begin{array}{c} s_k \\ v_k \end{array} \right)
=
\left( \begin{array}{cc} s_1 & -z v_1 \\ v_1 & \tilde s_1 \end{array} \right)
\left( \begin{array}{c} s_{k-1} \\ v_{k-1} \end{array} \right)
=
\left( \begin{array}{cc} s_1 & -z v_1 \\ v_1 & \tb s_1 \end{array} \right)^k
\left( \begin{array}{c} 1 \\ 0 \end{array} \right).
\end{equation}
\item[(iv)] For any $\theta\in[0,2\pi)$, let
\begin{equation}\label{antistretched}
\omega(\theta)= 2\arccos (\rho \cos\tfrac{\theta}{2}) \in [\phi,2\pi-\phi),
\end{equation}
where $\phi$ is~\eqref{eq:phi}.
If $z = e^{{\rm i}\theta}$, then
\begin{align}
\label{sExact} s_k(z) &=
 \cos k \omega-{\rm i}  \frac{\sin k\omega}{\sin\tfrac{\omega}{2}} \rho \sin\tfrac{\theta}{2} , \\
\label{vExact} v_k(z) &= - \frac{\sin k \omega}{\sin\tfrac{\omega}{2}} |\alpha| e^{-{\rm i}  \theta/2}.
\end{align}
for any $k\in\bbZ$.
\end{itemize}
\end{lemma}
\begin{proof}
(i), (ii), and (iii) for $k=0$ and $k=1$ can be checked directly. Combining $\calD^{k+1} = \calD \,\calD^{k}$,~\eqref{eq:powersD},~\eqref{vk}, 
and Lemma~\ref{lemDT}(i), we get
\begin{align}
\label{tempo0} s_{k+1} & =  s_1 s_{k} - \tb v_1  v_{k}, \\
\label{tempo1} v_{k+1} & = \tb s_1 v_{k} + v_1 s_{k}.
\end{align}
An easy induction proves that each $s_k$ and $v_k$ has real coefficients, that is,~\eqref{skvk2} and~\eqref{skvk3} hold.

Using $\calD^{k+1} = \calD^{k}  \, \calD$, ~\eqref{eq:powersD},~\eqref{vk}, and Lemma~\ref{lemDT}(i), 
we also get
\begin{equation}\label{tempo2}
v_{k+1} = \tb s_{k} v_{1} + v_{k} s_{1}.
\end{equation}
Equating the right-hand sides of~\eqref{tempo1} and~\eqref{tempo2}, we obtain
$$
v_{k} = \frac{v_1(\tb s_{k} - s_{k})}{\tb s_1 - s_1},
$$
which reduces to~\eqref{skvk1}. Then~\eqref{skvk4} follows immediately.

The recurrence in (iii) is just~\eqref{tempo0} and~\eqref{tempo1} rewritten in the matrix form after an application of~\eqref{skvk2},~\eqref{skvk3},~\eqref{skvk4}.

Let us prove (iv) now. Denote
$A(z) = \bigl(\begin{smallmatrix}
s_1&-z v_1 \\ v_1 & s_1^*
\end{smallmatrix} \bigr),$
the transfer matrix in~\eqref{recurrence}.
Using ~\eqref{s1},~\eqref{v1}, it is easy to see that 
$A(z)$ is unitary if $z = e^{{\rm i}\theta}$. The eigenvalues of $A(z)$ can be seen to be
\begin{equation}
\lambda_{\pm} = \frac{s_1 + s_1^*}{2} \pm \sqrt{\Big(\frac{s_1 + s_1^*}{2}\Big)^2-1}.
\end{equation}
Here $(s_1+s_1^*)/2 = -|\alpha|^2+\rho^2 \cos\theta \in [-1,0)$, and we adopt the convention that $\sqrt{(\tfrac{s_1 + s_1^*}{2})^2-1}$ belongs to ${\rm i}\bbR_+$ when $\theta\in [0,\pi)$ and to $-{\rm i}\bbR_+$ when $\theta\in(\pi,2\pi)$.

With this in mind, it is easy to see  that $\lambda_- = 1/\lambda_+$ and if $\lambda_+ = e^{{\rm i}\omega_+}$ with $\omega_+\in[0,2\pi)$, then
\begin{equation}\label{stretching}
\cos^2 \tfrac{\omega_+}{2}  = \frac{\cos \omega_+ + 1}{2 }= \frac{\lambda_+ + \lambda_- + 2}{4} = {\frac{-|\alpha|^2 + \rho^2 \cos\theta+1}{2}} = \rho^2 \cos^2\tfrac\theta2.
\end{equation}
Tracing the signs of cosines carefully, we can see that $\cos \frac{\omega_+}{2} = \rho \cos\frac\theta2$, that is, $\omega=\omega_+$, see~\eqref{antistretched}.

%

Let us now find the eigenvectors: let $\bigl(\begin{smallmatrix} x \\1 \end{smallmatrix} \bigr)$ and $\bigl(\begin{smallmatrix} -1 \\\bar x \end{smallmatrix} \bigr)$ be (orthogonal) eigenvectors of  $A(z)$ corresponding to the eigenvalues $\lambda_+$ and $\lambda_-$, respectively. Then
\begin{align*}
v_1 x+s_1^* &=\lambda_+,\\
-v_1 + s_1^* \bar x & = \lambda_- \bar x.
\end{align*}
These imply $x=(\lambda_+ - s_1^*)/v_1$ and $\bar x = v_1/(s_1^* - \lambda_-)$, which produce
\begin{align*}
\frac{x \bar x}{1+ x\bar x} &= \frac12 + \frac12 \frac{\lambda_+ +\lambda_- -2s_1^*}{\lambda_+ - \lambda_-}
= \frac{1}{2}-  \frac{\rho^2}{2} \frac{z-1/z}{\lambda_+ - 1/\lambda_+}
= \frac{1}{2}-  \frac{\rho}{2} \frac{\sin\tfrac{\theta}{2}}{\sin\tfrac{\omega}{2}}, \\
\frac{1}{1+ x\bar x} &= 1 -\frac{x \bar x}{1+ x\bar x} =  \frac{1}{2}+  \frac{\rho}{2} \frac{\sin\tfrac{\theta}{2}}{\sin\tfrac{\omega}{2}}.
\end{align*}


Writing $A(z)^k$ as
$
\bigl(\begin{smallmatrix}
x&-1 \\ 1 & \bar x
\end{smallmatrix} \bigr)
\bigl(\begin{smallmatrix}
\lambda^k_+ & 0 \\ 0 & \lambda_-^k
\end{smallmatrix} \bigr)
\bigl(\begin{smallmatrix}
x&-1 \\ 1 & \bar x
\end{smallmatrix} \bigr)^{-1},
 $
recurrence~\eqref{recurrence} gives
\begin{equation*}
s_k = \frac{|x|^2}{1+|x|^2} \lambda_+^k + \frac{1}{1+|x|^2} \lambda_-^k
\end{equation*}
which becomes~\eqref{sExact}. Note that for negative $k$ the formula~\eqref{sExact} holds since $s_{-k} = s_k^*$. Formula~\eqref{vExact} for $v_k$ follows from the one for $s_k$ and~\eqref{skvk1}.
\end{proof}

\begin{corollary}\label{cor:hcqt}
Let $h(z)= \sum_{j=-N}^N h_jz^j$ be a Laurent polynomial and set \begin{align}
\label{defS} S(z) &= \sum_{j=-N}^N h_j s_j(z) = h_0 + \sum_{j=1}^N (h_j s_j(z) + h_{-j} \tb s_j(z)), \\
\label{defV} V(z) &= \sum_{j=-N}^N h_j v_j(z) = \sum_{j=1}^N (h_j - h_{-j}) v_j(z).
\end{align}
Then
\begin{equation} \label{eq:hdinqt}
 Rh(\calC ) R^*= h(\calD)\stackrel{f.e.}{=} \QT(\bar S, \frac{|\alpha|}{\bar \alpha}V, \frac{|\alpha|}{\bar \alpha} \bar V, S).
 \end{equation}
\end{corollary}
\begin{proof}
Note that  
\begin{equation*}
h(\mathcal D) \fr \sum_{j=-N}^N h_j \QT(s_j,\tfrac{|\alpha|}{\bar \alpha} v_j, \tfrac{|\alpha|}{\bar \alpha} v_j,s_j)
 = \QT\Big( \sum_{j=-N}^N \bar h_j s_j , \tfrac{|\alpha|}{\bar \alpha} \sum_{j=-N}^N { h}_j v_j , \tfrac{|\alpha|}{\bar \alpha} \sum_{j=-N}^N \bar h_j v_j , \sum_{j=-N}^N { h}_j s_j\Big).
\end{equation*}
Since $s_j$ and $v_j$ have real coefficients, we obtain that the above equality can be rewritten as in \eqref{eq:hdinqt}.
\end{proof}

%

\subsection{Construction of $L$ and $U$}

We now come to the construction of $L$ and $U$ in Proposition~\eqref{prop:commutator}. We do this by decomposing each $\mathcal D^k$ into an upper- and lower-triangular part with a careful choice of diagonals.

First we introduce some notations.
For a Laurent polynomial $l(z) = \sum_{j=-q}^p l_j z^j$ let us define
\begin{align}
& l^+(z)  :=  \sum_{j=1}^p l_j z^j,  \qquad l^-(z)  :=  \sum_{j=-q}^{-1} l_j z^j, \\
& l^{+,\circ}(z)  :=  \sum_{j=0}^p l_j z^j, \qquad l^{-,\circ}(z)  :=  \sum_{j=-q}^0 l_j z^j, \qquad
l^{\circ}(z)  :=  l_0.
\end{align}
For a future reference we note that~\eqref{skvk4} and~\eqref{skvk3}   imply
\begin{align}
\label{ttilde1} (v_k^{+,\circ}) \tb{} & = z (v_k^-) ,\\
\label{ttilde2} (v_k^{-}) \tb{} & = z (v_k^{+,\circ}).
\end{align}

We are now ready to define
\begin{align}
\label{ekplus}
\calE_k^+ & := \QT(s_k^-,\tfrac{|\alpha|}{\bar \alpha}  v_k^{+,\circ}, \tfrac{|\alpha|}{\bar \alpha}  v_k^-, s_k^+), \\
\calE_k^- & := \QT(s_k^+,\tfrac{|\alpha|}{\bar \alpha}  v_k^{-}, \tfrac{|\alpha|}{\bar \alpha}  v_k^{+,\circ}, s_k^-), \\
\calE_k^\circ & := \QT(s_k^\circ,0,0,s_k^\circ)
\end{align}
for any $k\in\bbZ$. For a future reference, it is important for us that
$R^* \calE_k^+ R$ is lower triangular, $R^* \calE_k^- R$ is upper triangular, and $R^* \calE_k^\circ R$ is diagonal.

Trivially, $\calD^k \fr 
\calE_k^+ +\calE_k^- +\calE_k^\circ$, see~\eqref{eq:powersD}.

\begin{lemma}\label{lemComm0}
For any $k\in\bbZ$, $[\calD,\calE_k^+] \fr \QT(0,l_k,-l_k,0)$, where $l_k = \tfrac{\rho}{\bar{\alpha}}  v_1 v_k^\circ $. 
\end{lemma}
\begin{proof}
Define $q_1,q_2,q_3,q_4$ by
\begin{multline*}
\QT(q_1,q_2,q_3,q_4) \fr [\calD,\calE_k^+] \\
= \QT(s_1,\tfrac{|\alpha|}{\bar \alpha}  v_1,\tfrac{|\alpha|}{\bar \alpha}  v_1,s_1) \QT(s_k^-,\tfrac{|\alpha|}{\bar \alpha}  v_k^{+,\circ}, \tfrac{|\alpha|}{\bar \alpha}  v_k^-, s_k^+) \\
- \QT(s_k^-,\tfrac{|\alpha|}{\bar \alpha}  v_k^{+,\circ}, \tfrac{|\alpha|}{\bar \alpha}  v_k^-, s_k^+)
\QT(s_1,\tfrac{|\alpha|}{\bar \alpha}  v_1,\tfrac{|\alpha|}{\bar \alpha}  v_1,s_1) .
\end{multline*}
Using Lemma~\ref{lemDT}(i) and then~\eqref{skvk3},~\eqref{skvk4},~\eqref{ttilde1}, we get
$$
q_1 = s_1 s_k^-  - (\tfrac{|\alpha|}{\bar \alpha}  v_1)\tb{}  \tfrac{|\alpha|}{\bar \alpha}  v_k^-           - s_k^- s_1 + (\tfrac{|\alpha|}{\bar \alpha}  v_k^{+,\circ}) \tb{}  \tfrac{|\alpha|}{\bar \alpha}  v_1
= - z  v_1 v_k^-     + z  v_k^{-} v_1 = 0.
$$

For $q_2$ let us again use Lemma~\ref{lemDT}(i) to get
\begin{equation*}
q_2 = \tb s_1 \tfrac{|\alpha|}{\bar \alpha}  v_k^{+,\circ}  + \tfrac{|\alpha|}{\bar \alpha}  v_1  s_k^+
- (s_k^-)\tb{}   \tfrac{|\alpha|}{\bar \alpha} v_1 - \tfrac{|\alpha|}{\bar \alpha}  v_k^{+,\circ} s_1
=
\tfrac{|\alpha|}{\bar \alpha}  v_1 (s_k^{+} - (s_k^-) \tb{} ) - \tfrac{|\alpha|}{\bar \alpha}  v_k^{+,\circ} ( s_1 - \tb s_1).
\end{equation*}
Now let us rewrite~\eqref{skvk1} as
\begin{equation}\label{skvk1modif}
v_k -  z v_k = -\tfrac{|\alpha|}{\rho} ( s_k - \tb s_k),
\end{equation}
 and project onto the positive powers of $z$:
$$
v_k^{+} - z v_k^{+,\circ}  = -\tfrac{|\alpha|}{\rho}( s_k^{+} - (s_k^-)\tb{}  ) ,
$$
These two equalities allow us to rewrite the expression for $q_2$ as
$$
q_2 = -\tfrac{\rho}{\bar{\alpha}} v_1 (v_k^{+} - z v_k^{+,\circ}  ) + \tfrac{\rho}{\bar{\alpha}} v_k^{+,\circ} (v_1 - z v_1)
= \tfrac{\rho}{\bar{\alpha}}  v_1 (v_k^{+,\circ} - v_k^+ ) = l_k
.
$$

Similarly, if we project~\eqref{skvk1modif} onto the negative powers of $z$ and use~\eqref{skvk4}, we get
$$
v_k^{-} -  (v_k^+)\tb{} = -\tfrac{|\alpha|}{\rho} ( s_k^{-} - (s_k^+)\tb{} ).
$$
This together with Lemma~\ref{lemDT}(i) and~\eqref{ttilde1} allows to compute $q_3$:
\begin{align*}
q_3 & = \tfrac{|\alpha|}{\bar \alpha} v_1 s_k^{-}  +\tb s_1  \tfrac{|\alpha|}{\bar \alpha} v_k^{-}            - \tfrac{|\alpha|}{\bar \alpha} v_k^{-} s_1 - (s_k^{+})\tb{} \tfrac{|\alpha|}{\bar \alpha} v_1
\\
&= \tfrac{|\alpha|}{\bar \alpha} v_1 (s_k^{-} -(s_k^{+})  \tb{} ) - \tfrac{|\alpha|}{\bar \alpha} v_k^{-} (s_1 - \tilde s_1 )
\\
&= -\tfrac{\rho}{\bar{\alpha}} v_1 ( z v_k^{-} -(v_k^{+})\tb{} )  = -\tfrac{\rho}{\bar{\alpha}} v_1 ( (v_k^{+,\circ})\tb{}  -(v_k^{+}) \tb{} )  = - l_k.
\end{align*}

Finally, we compute $q_4$:
\begin{equation*}
q_4 = -(\tfrac{|\alpha|}{\bar \alpha}  v_1)\tb{}  \tfrac{|\alpha|}{\bar \alpha}  v_k^{+,\circ}   +s_1   s_k^{+}                +
(\tfrac{|\alpha|}{\bar \alpha}  v_k^-)\tb{}   \tfrac{|\alpha|}{\bar \alpha}  v_1 - s_k^{+} s_1
= - \tilde v_1  v_k^{+,\circ} + (v_k^{-})\tb{}  v_1=- z v_1  v_k^{+,\circ} + z v_k^{+,\circ} v_1 = 0,
\end{equation*}
where we have used~\eqref{ttilde2}.
\end{proof}

\begin{lemma}
For any Laurent polynomials $s,t,u_1,u_2$ we have
$$
[\DT(s,t),\QT(u_1,0,0,u_2)] \fr \QT(0, t(u_2- \tb u_1), t(u_1- \tb u_2) ,0).
$$
\end{lemma}
\begin{proof}
Using Lemma~\ref{lemDT}(i), we get
\begin{multline*}
[\DT(s,t),\QT(u_1,0,0,u_2)]  = \QT(s,t,t,s) \QT(u_1,0,0,u_2) - \QT(u_1,0,0,u_2) \QT(s,t,t,s)
\\
\fr \QT(s u_1 - u_1 s, t u_2 - \tb u_1 t, t u_1 - \tb u_2 t, s u_2 - u_2 s).
\end{multline*}
\end{proof}

Let us now modify $\calE_k^+$ to make the commutator in Lemma~\ref{lemComm0} of trace class:
\begin{equation}\label{ekplushat}
\hat \calE_k^+ : = \calE_k^+ + \QT(  \tfrac{\rho}{|\alpha|} v_k^\circ,0,0,0).
\end{equation}

\begin{lemma}\label{lemEkplus}
For any $k,j\in\bbZ$ we have that $[\calD^j,\hat \calE_k^+ ] \fr \bdnot$. 
\end{lemma}
\begin{proof}
Using the previous two lemmas:
\begin{equation*}
[\calD,\hat \calE_k^+ ] = [\calD, \calE_k^+ ] + [\calD,\QT(  \tfrac{\rho}{|\alpha|} v_k^\circ,0,0,0)]
\fr \QT(0,l_k,-l_k,0) + \QT(0, -\tfrac{|\alpha|}{\bar \alpha} v_1 (\tfrac{\rho}{|\alpha|} v_k^\circ)^*,  \tfrac{|\alpha|}{\bar \alpha} v_1 \tfrac{\rho}{|\alpha|} v_k^\circ,0 ) = \bdnot.
\end{equation*}

The statement for general $j\in\bbZ$ follows from
$$
[\calD^{-1},  \hat \calE_k^+] = - \calD^{-1} [\calD, \hat \calE_k^+] \calD,
$$
and
$$
[\calD^j,  \hat \calE_k^+]  = \calD [\calD^{j-1}, \hat \calE_k^+ ] +
[\calD ,\hat \calE_k^+] \calD^{j-1},
$$
together with an induction in $j$.
\end{proof}
We are now ready for the construction of $L$ and $U$.

\begin{lemma}\label{lemComm}
Let $h(z) = \sum_{j=-N}^N h_j z^j$ a Laurent polynomial and define $S$ and $V$ as in \eqref{defS} and \eqref{defV}, so that
$$h(\calC)= R^*  \QT(\bar S, \tfrac{|\alpha|}{\bar \alpha} V,\tfrac{|\alpha|}{\bar \alpha} \bar V,S)R + F$$
for a matrix $F$ with $F\fr \bdnot$, see Corollary \ref{cor:hcqt}.

Then, with  
\begin{equation}\label{eq:defL}
L =R^* \QT(\bar S^{-} + \tfrac{\rho}{|\alpha|} \bar V^\circ, \tfrac{|\alpha|}{\bar \alpha} V^{+,\circ},\tfrac{|\alpha|}{\bar \alpha} \bar V^{-},S^{+})R^*+ F_{LT},
\end{equation}
where $F_{LT}$ is the strictly lower triangular part of $F$, and $U=h(\calC)-L$, we have that $[U,L]$ is trace class and
\begin{equation}\label{CommTrace}
\Tr [U,L] = 2 \sum_{j=1}^N j S_j S_{-j} -\sum_{j=1}^N j (V_j^2 + V_{-j}^2),
\end{equation}
where $S_j$'s and $T_j$'s are the Laurent coefficients of $S$ and $T$, respectively: $S(z) = \sum S_j z^j$, $T(z) = \sum T_j z^j$.
\end{lemma}
\begin{proof}
First, note that $U$ is upper triangular and $L$ is lower triangular.
Since $F\fr \bdnot$, we also have $F_{LT}\fr \bdnot$. 
 Now note that  $[U,L]= [h(\mathcal C),L]$.
The latter is sum over terms including $F_{LT}$ and  $R^*[\calD^j,\hat{\mathcal E}^+_k ]R$
and hence we see that $[U,L]$ is also  finite rank by Lemma~\ref{lemEkplus}, and hence in particular trace class. Therefore we can apply Lemma~\ref{prop:commutator}.

Using the fact that the trace is invariant under unitary conjugation and that $\Tr[A,B]=0$ if either $A$ or $B$ is finite rank, we find
\begin{equation}
\Tr[U,L]
= \Tr \left[ \QT(\bar S , \tfrac{|\alpha|}{\bar \alpha} V,\tfrac{|\alpha|}{\bar \alpha} \bar V,{S}), \QT(\bar S^{-} + \tfrac{\rho}{|\alpha|} \bar V^\circ, \tfrac{|\alpha|}{\bar \alpha} V^{+,\circ},\tfrac{|\alpha|}{\bar \alpha} \bar V^{-},{S}^{+})\right].
\end{equation}
From the definition of $\DT$ and $\QT$, 
we get
\begin{multline*}
\Tr [U,L] = \Tr \Big( \T(\bar S)^*\T(\bar S^{-} + \tfrac{\rho}{|\alpha|} \bar V^\circ)^* - \T(\tfrac{ |\alpha|}{\bar \alpha} V)\T( \tfrac{ |\alpha|}{\bar \alpha} \bar V^- )^* \Big.
 - \T( \tfrac{|\alpha|}{\bar \alpha} \bar V)^* \T(\tfrac{|\alpha|}{\bar \alpha} V^{+,\circ}) + \T( S ) \T(S^{+}) \\
- \T(\bar S^{-} + \tfrac{\rho}{|\alpha|} \bar V^\circ)^* \T( \bar S)^* + \T(\tfrac{|\alpha|}{\bar \alpha} V^{+,\circ})\T(\tfrac{|\alpha|}{\bar \alpha} \bar V)^*
\Big. + \T( \tfrac{ |\alpha|}{\bar \alpha} \bar V^- )^*  \T(\tfrac{ |\alpha|}{\bar \alpha} V) - \T(S^{+}) \T( S ) \Big) .
\end{multline*}
Note that $T(\bar s)^* = T(\tilde s)$. Finally, by \eqref{hankel} we obtain
\begin{multline*}
\Tr [U,L] =
\Tr( \T(0)) - \Tr\Big( \Ha(\tilde S)\Ha( S^{-} + \tfrac{\rho}{|\alpha|}  V^\circ) - \Ha(\tfrac{ |\alpha|}{\bar \alpha} V) \Ha(\tfrac{ |\alpha|}{ \alpha} V^-) \Big.
- \Ha( \tfrac{|\alpha|}{ \alpha}  V) \Ha(\tfrac{|\alpha|}{\bar \alpha} (V^{+,\circ})\,\tilde{}\,\,) \\
+ \Ha(S)\Ha((S^+)\,\tilde{}\,\,)
- \Ha( (S^{-})\,\tilde{}\,\, + \tfrac{\rho}{|\alpha|}  V^\circ) \Ha( S) + \Ha(\tfrac{|\alpha|}{\bar \alpha} V^{+,\circ})\Ha(\tfrac{|\alpha|}{ \alpha} V)
\Big. + \Ha( \tfrac{ |\alpha|}{ \alpha} (V^-)\,\tilde{}\,\, )  \Ha(\tfrac{ |\alpha|}{\bar \alpha} \tilde V) - \Ha(S^{+}) \Ha( \tilde S ) \Big) .
\end{multline*}
 Now observe that $\Tr \big(\Ha(m) \Ha( n) \big) = \sum_{j=1}^\infty j m_j n_{j}$. Applying this to each term in the previous expression leads to
\begin{equation*}
\Tr [U,L] =
\sum_{j=1}^\infty j S_{-j} S_{j}-  \sum_{j=1}^\infty  j V_j^2 - \sum_{j=1}^\infty  j V_{-j}^2  + \sum_{j=1}^\infty  j S_{j} S_{-j} .
\end{equation*}
This finishes the proof.
\end{proof}
\subsection{Proof of Theorem \ref{thm1}}
We are almost done with the proof of Theorem \ref{thm1}. It remains to put \eqref{CommTrace} in the form \eqref{eq:defQa} and extend the results for Laurent polynomials to $h \in \mathfrak B_{\frac12}$.

We need one final lemma.

\begin{lemma} \label{lem:final}  The maps $h \mapsto \mathcal A^h$ and $h \mapsto \mathcal B^h$ (see~\eqref{eq:Ah}, ~\eqref{eq:Bh}) satisfy the following properties:
\begin{enumerate}

\item[(i)] If $h$ is a Laurent polynomial of degree $N$, then $\Tr[U,L]$ from Lemma~\ref{lemComm} can be written as
\begin{equation}\label{CommTrace2}
\Tr[U,L] =  2 \sum_{j=1}^N j \mathcal{A}^h_j \mathcal{A}^h_{-j}+
2 \sum_{j=1}^N j \mathcal{B}^h_j \mathcal{B}^h_{-j}.
\end{equation}
Here $\mathcal{A}^h_j$ and $\mathcal{B}^h_j$ are the $j$-th Fourier coefficients \eqref{eq:fourier} of $\mathcal{A}^h$ and $\mathcal{B}^h$.
\item[(ii)] If $h$ is Laurent polynomial of degree $N$, then also $\mathcal A^h$ and $\mathcal B^h$ are Laurent polynomials of degree $N$.
\item[(iii)]   For $h \in \mathfrak B_{\frac12}$ we have
$$\sum_{j=1}^\infty j |\mathcal A^h_j|^2  \leq c\|h\|_{\mathfrak B_{\frac12}}, \qquad \sum_{j=1}^\infty j |\mathcal B^h_j|^2  \leq c\|h\|_{\mathfrak B_{\frac12}},$$
for some constant $c>0$.
\end{enumerate}

\end{lemma}
\begin{remark}
Clearly $\mathcal{A}^h(z)$ is $h^{ev}(w(\theta))$, where $h^{ev}$ is the even part of $h$. Up to a prefactor, $\mathcal{B}^h(z)$ depends on the odd part of $h$ similarly. 
\end{remark}
\begin{proof}
(i) Define the Laurent polynomials
\begin{align}
W(z) & = 
2 \sin\tfrac\theta2   \sum_{j=1}^N b_j  \frac{\sin  j\omega}{\sin \frac{\omega}{2}}, \\
R(z) & = 2 \cos\tfrac\theta2   \sum_{j=1}^N b_j  \frac{\sin  j\omega}{\sin \frac{\omega}{2}}.
\end{align}
Then~\eqref{defS} and~\eqref{defV} imply that
\begin{align}
S(z) & = \mathcal{A}^h(z) + \rho W(z), \\
V(z) & = -|\alpha| W(z) - {\rm i} |\alpha| R(z).
\end{align}
Note that $W_{-j} = -W_j$ and $R_{-j} = R_j$. This gives
$$
\sum_{j=1}^N j S_j S_{-j} = \sum_{j=1}^N j \mathcal{A}^h_j \mathcal{A}^h_{-j} - \rho^2 \sum_{j=1}^N j W_j^2,
$$
and
$$
\sum_{j=1}^N j  (V_j^2 + V_{-j}^2) = 2|\alpha|^2 \sum_{j=1}^N j W_j^2 - 2|\alpha|^2 \sum_{j=1}^N j R_j^2.
$$
Therefore, by Lemma~\ref{lemComm}
\begin{equation}
\Tr[U,L]  =  2  \sum_{j=1}^N j \mathcal{A}^h_j \mathcal{A}^h_{-j} -2  \sum_{j=1}^N j W_j^2  + 2|\alpha|^2 \sum_{j=1}^N j R_j^2,
\end{equation}
which equals to~\eqref{CommTrace2} since $\mathcal{B}^h(z) = W(z) + |\alpha| R(z)$.

(ii) Recall that the Chebyshev polynomials $T_n$ and $U_n$ of the first and the second kind can be defined via
\begin{align}
T_n(\cos t) & = \cos nt, \\
U_n(\cos t) & = \frac{\sin (n+1)t}{\sin t}.
\end{align}
Using these, one can rewrite $\mathcal{A}^h$ and $\mathcal{B}^h$:
\begin{align}
\label{eq:ChebEven}
\mathcal{A}^h(z) & = a_0+2\sum_{j=1}^N a_j  T_{2j}(\rho \cos \tfrac\theta2), \\
\label{eq:ChebOdd}
\mathcal{B}^h(z) & = 2 (\sin\tfrac\theta2 + |\alpha| \cos\tfrac\theta2)  \sum_{j=1}^N b_j   U_{2j-1}(\rho \cos \tfrac\theta2).
\end{align}
Representations~\eqref{eq:ChebEven} and~\eqref{eq:ChebOdd} show that if $h$ is a Laurent polynomial in $z$ then so are $\mathcal{A}^h(z)$ and $\mathcal{B}^h(z)$. Indeed, $T_{2j}(x)$ contains only even powers of $x$, so using $\cos^2\tfrac\theta2 = \tfrac{\cos\theta+1}{2}$, is follows that $\mathcal{A}^h(z)$ can be represented as a linear combination of $\cos^k\theta$ for $k\ge 0$. This shows that $\mathcal{A}^h(z)$ is a Laurent polynomial in $z$. Similarly, $U_{2j-1}(x)$ contains only odd powers of $x$, so using $\cos^2\tfrac\theta2 = \tfrac{\cos\theta+1}{2}$ and $\sin \tfrac\theta2 \cos \tfrac\theta2 = \tfrac{\sin\theta}{2}$, we obtain that $\mathcal{B}^h(z)$  is a Laurent polynomial in $z$.

(iii)  From (ii) we know that $z \to \cos j \omega $ is a Laurent polynomial of degree $j$. Hence, the Hankel matrix $H(\cos j \omega)$ is of rank $j$ and thus
$$\|H(\cos j \omega)\|_2  \leq \sqrt{j} \|H(\cos j \omega)\|_\infty \leq \sqrt{j}. $$
Then
\begin{equation*}
\left(\sum_{j=1}^\infty j |\mathcal A_j^h|^2\right)^{1/2}= \|H(\mathcal A^h)\|_2\leq 2 \sum_{j}|h_j| \|H(\cos j \omega)\|_2
\leq 2 \sum_{j}\sqrt{j} |h_j|  \leq 2 \|h\|_{ \mathfrak B _{\frac12}}.
\end{equation*}
The case of $\mathcal B^h$ is identical.

\end{proof}

We are now ready for the proof of  Theorem \ref{thm1}.

\begin{proof}[Proof of Theorem \ref{thm1}]
We first  assume again that $h$ is real-valued.

By Theorem \ref{thm2} we only have the consider the case $\alpha_{n}\equiv \alpha$. In that case, if $h$ is a Laurent polynomial, then the statement follows from combining Proposition \ref{prop:commutator}, Lemma \ref{lemComm} and the first property in Lemma \ref{lem:final}.

The extension to real-valued $h \in \mathfrak B_{\frac12}$ goes by a standard continuity argument. Let $\{h_N\}$ be a sequence of Laurent polynomials converging to $h$ in $\mathfrak B_{\frac12}$. Then from \eqref{eq:equicontinuitysobolev} and the continuity of $h\mapsto Q_\alpha (h)$ ensured by the third property in Lemma \ref{lem:final} we obtain
$$
\lim_{n\to \infty} \Psi_n(h,\mathcal C)= \lim_{N\to \infty} \lim_{n\to \infty} \Psi_n(h_N,\mathcal C)= \lim_{N\to \infty} {\rm e}^{Q_\alpha(h_N)}={\rm e}^{Q_\alpha(h)},
$$
and this proves the statement.

Finally, for complex-valued functions $h$, we can use  a normal family argument very similar to the ones we used in the proofs of Theorem \ref{thm2}, \ref{thm:rightlimit} and Proposition \ref{prop:weak}. (Observe that $z \to Q_\alpha(h_z)$ is quadratic.) We leave the details to the reader.
\end{proof}
\subsection{Proof of Proposition \ref{prop:alpha}}
\begin{proof}[Proof Proposition \ref{prop:alpha}]
Again by a normal family argument using \eqref{eq:bounddervativephiincommutators} we find that there exists a subsequence $\{n_j\}_{j \in \mathbb N}$ of $\mathbb N$ such that $\Psi_{n_j}(h,\mu)$ converges and we denote the limit by $\Psi(h,\mu)$. We are done if we  show that the limit does not depend on the subsequence and is always given by $\exp(Q_\alpha(h))$.

 Now let $\{\beta_k\}_{k \in \mathbb N}$ be a right limit of $\{\alpha_{n_j}\}_{j \in \mathbb N}$ along $\{j_\ell\}_{\ell \in \mathbb N}$. Then by \eqref{L1} we know that $\beta_k= |\alpha| {\rm e}^{{\rm i} \phi_k}$ for some angles $\phi_k$. Then, by \eqref{L2}, we find that $\phi_k= \phi$ is independent of $k$. Hence $\beta_k=|\alpha| {\rm e}^{{\rm i } \phi}$ for $k \in \mathbb Z$. Since $\{\beta_k\}_{k \in \mathbb N}$ is of course also a right limit of the full sequence $\{\alpha_n\}_{n \in \mathbb N}$ along $\{n_{j_\ell}\}_{\ell \in \mathbb N}$, we can apply Theorem \ref{thm2}.  This proves that $\Psi_{n_{j_\ell}}$  converges to $\exp( Q_{\alpha}(h))$ as $\ell \to \infty$, as it does not depend on the phase $\phi$. Hence $\Psi_{n_j}$ converges to $\exp( Q_{\alpha}(h))$ as $j \to \infty$ and this proves the statement.
\end{proof}
\appendix
\section{Constant Verblunsky coefficients}

If  ${\rm d} \mu= \frac{{\rm d} \theta}{2 \pi}$ then $\Phi_n(z)= z^n$ and $\alpha_n\equiv 0$.  Another example that is of particular interest to us is the case $\alpha_n\equiv \alpha$ with $|\alpha|<1$ (the corresponding orthogonal polynomials bear the name of Geronimus polynomials, see~\cite{Ger} and~\cite[Ex. 1.6.12]{OPUC1}). The associated measure $\mu_\alpha$ is given by
\begin{equation}
\label{eq:muAlpha}
{\rm d}\mu_\alpha(\theta) = w(\theta) \tfrac{{\rm d}\theta}{2\pi} + q \, \delta_\beta(\theta),
\end{equation}
where
the a.c. part of $\mu_\alpha$ is
$$
w(\theta) =
\begin{cases}
\tfrac{1}{|1+\alpha|} \frac{\sqrt{\cos^2(\phi/2)-\cos^2(\theta/2)}}{\sin((\theta-\beta)/2)} \tfrac{{\rm d}\theta}{2\pi} & \mbox{for } \theta\in(\phi,2\pi-\phi), \\
0 & \mbox{for } \theta\in[-\phi,\phi],
\end{cases}
$$
where
$$
\phi = 2\arcsin(|\alpha|),
$$
and $\beta$ is defined from
$$
1+\bar{\alpha} = |1+\bar{\alpha}| \exp(i\beta/2).
$$
The singular part of $\mu_\alpha$ may consist of up to one pure point located at $e^{i\beta}$ with the weight
$$
\mu_\alpha(\{e^{i\beta}\}) = q =
\begin{cases}
 0 & \mbox{if } |\alpha+\tfrac12|\le \tfrac12, \\
\tfrac{2}{|1+\alpha|^2}(|\alpha+\tfrac12|^2 - \tfrac14) & \mbox{if } |\alpha+\tfrac12| > \tfrac12.
\end{cases}
$$

\end{document}